\documentclass[pdflatex,sn-mathphys-num]{sn-jnl}% Math and Physical Sciences Numbered Reference Style
\usepackage{graphicx}%
\usepackage{multirow}%
\usepackage{amsmath,amssymb,amsfonts}%
\usepackage{amsthm}%
\usepackage{mathrsfs}%
\usepackage[title]{appendix}%
\usepackage{xcolor}%
\usepackage{textcomp}%
\usepackage{manyfoot}%
\usepackage{booktabs}%
\usepackage{algorithm}%
\usepackage{algorithmicx}%
\usepackage{algpseudocode}%
\usepackage{listings}%

\usepackage{shortcuts}
\usepackage{array}
\usepackage{ragged2e}
\usepackage{threeparttable}

\theoremstyle{thmstyleone}%
\newtheorem{theorem}{Theorem}%  meant for continuous numbers
\theoremstyle{thmstyletwo}%

\theoremstyle{thmstylethree}%

\begin{document}

\title[Fast Uncertainty Quantification for Kernel-Based Estimators]{Fast Uncertainty Quantification for Kernel-Based Estimators in Large-Scale Causal Inference}

\author*[1]{\fnm{Matthew} \sur{Kosko}}\email{Matthew.Kosko@nyulangone.org}

\author[2]{\fnm{Falco J.} \sur{Bargagli-Stoffi}}

\author[3]{\fnm{Lin} \sur{Wang}}

\author[1]{\fnm{Michele} \sur{Santacatterina}}

\affil*[1]{\orgdiv{Department of Population Health}, \orgname{New York University}, \orgaddress{\city{New York}, \state{NY}, \postcode{10016}, \country{USA}}}

\affil[2]{\orgdiv{Department of Biostatistics}, \orgname{University of California}, \orgaddress{\city{Los Angeles}, \state{CA}, \country{USA}}}

\affil[3]{\orgdiv{Department of Statistics}, \orgname{Purdue University}, \orgaddress{\city{West Lafayette}, \state{IN}, \postcode{47907}, \country{USA}}}


\abstract{Kernel methods are widely used in causal inference for tasks such as treatment effect estimation, policy evaluation, and policy learning. The bootstrap is a standard tool for uncertainty quantification because of its broad applicability. As increasingly large datasets become available, such as the 2023 U.S. Natality data from the National Vital Statistics System (NVSS), which includes 3,605,081 registered births (of which 2,980,739 form our analytic sample after cohort restrictions), the computational demands of these methods increase substantially. Kernel methods are known to scale poorly with sample size, and this limitation is further exacerbated by the repeated re-fitting required by the bootstrap. As a result, bootstrap-based inference for kernel-based estimators can become computationally infeasible in large-scale settings. In this paper, we address these challenges by extending the causal Bag of Little Bootstraps (cBLB) algorithm of \cite{kosko2023fast} to kernel methods; the method is the bag of little bootstraps in which each bootstrap replicate reweights per-unit contributions from estimators fitted once per subset, rather than refitting them, and the novelty here is its application, theoretical justification, and empirical evaluation for kernel-based estimators. Our approach achieves computational scalability by combining subsampling and resampling while preserving first-order uncertainty quantification and asymptotically correct coverage. We evaluate the method across three representative implementations: kernelized augmented outcome-weighted learning, kernel-based minimax weighting, and double machine learning with kernel support vector machines. We show in simulations that our method yields confidence intervals with nominal coverage at a fraction of the computational cost. We further demonstrate its utility in a real-world application by estimating the effect of any amount of smoking on birth weight, as well as the optimal treatment regime, using the NVSS dataset, where the standard bootstrap is prohibitively expensive computationally and effectively infeasible at this scale.}




\maketitle

\section{Introduction}
\label{sec:intro}

Kernel methods represent a versatile class of algorithms that enable powerful nonparametric modeling, and they have become essential tools across multiple areas of causal inference. Applications include modeling treatment effect heterogeneity \citep{ImaiRatkovic2013}, learning optimal treatment policies \citep{zhao2012estimating, zhou2017augmented, zhou2017residual, liu2018augmented}, estimating kernel minimax weights to correct for covariate imbalance \citep{kallus2021more, kallus2022optimal, hirshberg2021augmented}, estimating nuisance functions in doubly robust estimators \citep{naimi2017nonparametric}, and estimating dose-response curves \citep{singh2024kernel,kennedy2017non}. These methods operate by mapping data into high-dimensional feature spaces, enabling the use of linear operations to capture complex nonlinear relationships. We provide additional details on the central role of kernel methods in these approaches in Section~\ref{sec:rel_work}.

The bootstrap methodology \citep{eforn1979bootstrap,efron1994introduction} has been used to calculate precision estimates for statistical estimators, such as standard errors and confidence intervals. The literature on the bootstrap's use in causal inference, particularly for large datasets, is still relatively sparse (see \cite{kosko2023fast} and section \ref{sec:rel_work} below). 

The need to scale both estimation and inference methods is increasingly pressing, as researchers are collecting and analyzing ever-larger observational and experimental datasets to infer causal effects. For example, \cite{bargagli2020causal} and \cite{wu2024matching} used data on over 35 million Medicare beneficiaries across the US to study heterogeneous treatment effects and estimate average causal exposure-responses, respectively. \cite{nazaret2023large} analyzed 76 billion minutes of Apple Watch activity data from 160,000 users to evaluate a standing nudge (hourly alert to stand) using a regression discontinuity design (RDD). \cite{de2024ambient} conducted a multi-city Indian study on air quality and mortality, leveraging  3.6 million death records across 10 cities (2008--2019) using  instrumental variables (IV) to isolate the effect of locally generated PM2.5 pollution on deaths.

Despite their usefulness, neither kernel methods nor the bootstrap scale well to large datasets. Although increased computing power has made the bootstrap more accessible, it remains computationally burdensome when applied to large datasets, particularly in settings that require complex resampling procedures. To overcome these challenges, \cite{kleiner2014scalable} proposed the bag of little bootstraps (BLB), which, instead of applying an estimator to a smaller subsample, deploys the bootstrap on multiple subsets or ``bags'' of the data. Finally, BLB draws bootstrap samples equal to the size of the full dataset. While BLB improves the bootstrap's efficiency, it was not introduced as a method for quantifying the precision of causal effect estimators. \cite{kosko2023fast} introduced a new bootstrap algorithm called causal bag of little bootstraps (cBLB) for estimating average treatment effects with large data. The authors used a weighted version of the BLB to obtain estimates of uncertainty for variance-stabilized inverse probability weighting (IPW) estimators using weights obtained by logistic regression, covariate balancing propensity score (CBPS) \citep{imai2014covariate}, and machine learning. \cite{kosko2023fast} showed that their proposed method significantly enhances the computational efficiency of the traditional bootstrap while maintaining consistent estimation of the average treatment effect from large datasets.

Although recent advances, such as cBLB, offer promising results, they still inherit many limitations of traditional approaches. In particular, they are susceptible to bias under model misspecification \citep{kang2007demystifying}, may yield invalid inferences when using flexible, data-adaptive methods \citep{naimi2023challenges}, and can lead to erroneous inferences under lack of overlap \citep{kallus2021more}. In addition, the computational complexity only increases with the use of kernels. As is well-known, these kernel-methods are computationally-intensive with large datasets and there has been increasing interest in how to scale these methods \citep{kim2023scalable}

\subsection{Contribution}

Our work is motivated by the increasing usage of the bootstrap to construct confidence intervals in the context of causal inference with large real-world data using flexible kernel methods and accounts for the shortcoming of cBLB approaches. In particular, we contribute to this area of research a bootstrap procedure for kernel-based estimators that significantly improves the computational efficiency of the standard nonparametric bootstrap for calculating standard errors and confidence intervals. Our approach builds upon and extends the work of \cite{kleiner2014scalable} and \cite{kosko2023fast} on bag of little bootstrap algorithms for statistical and causal inference estimators. Whereas \cite{kosko2023fast} used a propensity-weighted BLB specific to variance-stabilized IPW, our method is the BLB in which each bootstrap replicate reweights per-unit contributions from estimators fitted once per subset, rather than refitting them; what is new here is extending it to, and establishing its first-order validity for, kernel-based estimators, whose per-subset fitting costs (e.g., kernel Gram matrices, kernel-weight optimization, kernel support vector machine (SVM) cross-fitting) make the standard bootstrap infeasible at scale.

Our approach, which keeps the name causal bag of little bootstraps (cBLB) from \cite{kosko2023fast}, replaces standard bootstrap refitting with a subsampling-and-resampling scheme tailored to kernel estimators. We first randomly partition the full dataset into $s$ disjoint bags of size $b \ll n$, with $b$ chosen so that kernel fitting remains computationally feasible. Within each bag, we fit the kernel-based estimator (and any required nuisance models or optimization problems) once, compute per-observation contributions that can be re-aggregated under resampling weights, and then hold these fitted objects fixed across the $r$ bootstrap replicates for that bag. These contributions may be influence-function/score terms for augmented inverse probability weighting (AIPW)/double machine learning (DML)-type estimators or weighted loss/residual terms for kernelized policy learning. We then generate $r$ bootstrap replicates by drawing multinomial counts $M_1,\ldots,M_b \sim \mathrm{Multinomial}\!\left(n;\frac{1}{b},\ldots,\frac{1}{b}\right)$ and forming the bootstrap replicate as the weighted average $\hat{\theta}^{*(\ell)}=\frac{1}{n}\sum_{a=1}^b M_a \hat{\theta}_{i_a}$, which mimics an $n$-out-of-$n$ nonparametric bootstrap without recomputing kernels or refitting models. Finally, we construct confidence intervals within each bag, for example percentile intervals across the $r$ replicates, and aggregate across bags by averaging the bag-specific quantiles, yielding a scalable approximation to standard bootstrap uncertainty.

Through Monte Carlo simulations, we show that the proposed method substantially improves computational efficiency while maintaining nominal finite-sample coverage in our main simulation settings, consistent with the first-order uncertainty quantification and asymptotically correct coverage established by our theory. We evaluate the method across three representative implementations: optimal policy learning via kernelized augmented outcome-weighted learning (AOL) \cite{liu2018augmented,zhou2017augmented}, kernel-based minimax weighting \cite{hirshberg2021augmented}, and double machine learning (DML) with kernel support vector machines (SVMs) \cite{kennedy2024semiparametric}. These three different methods represent a range of kernel-based approaches used in causal inference, and they allow us to demonstrate the versatility and robustness of our proposed method across different settings. Our method yields confidence intervals with nominal coverage at a fraction of the computational cost. For kernelized AOL, the target is the optimal value. Each subset learns its own treatment rule, which is then held fixed across that subset's bootstrap replicates, in the same way as the nuisance functions of the other estimators; the rule is not re-learned inside each replicate. The resulting interval targets the optimal value; it does not provide inference on the learned rule itself, or for a rule re-learned within each bootstrap replicate.
{In addition, we provide the \texttt{R} code for the algorithm at \url{https://github.com/mdk31/kernel-blb}.}

\subsection{Smoking and Birthweight in the NVSS: A Setting Where the Standard Bootstrap is Computationally Prohibitive}

We demonstrate our proposed method using the 2023 National Vital Statistics System (NVSS) natality dataset, which contains 3,605,081 registered births in the United States \citep{nchs_natality2023_userguide}, of which 2,980,739 singleton live births with complete covariate information form our analytic sample (Section~\ref{sec:studypop}). The dataset includes detailed information on maternal sociodemographic characteristics, prenatal care, health behaviors including trimester-specific smoking status, and infant outcomes including birthweight. This setting is compelling for two reasons. First, the well-characterized causal effect of smoking on birthweight provides a credible benchmark against which to evaluate our estimators: we expect negative effects of plausible magnitude (%maternal smoking during pregnancy is one of the most well-established modifiable risk factors for reduced infant birthweight, with decades of evidence consistently documenting 
research shows reductions on the order of 150--250 grams \citep{hellerstedt1995effects, wisborg2001maternal, adegboye2010maternal, cattaneo2010efficient}).  Second, the scale of the dataset is precisely the regime in which standard bootstrap inference becomes computationally infeasible, motivating the need for scalable resampling alternatives such as the one proposed.

In this motivating application, we estimate distinct causal quantities using our proposed methods. The kernel minimax weights and DML with SVMs approaches both estimate the average treatment effect of smoking during pregnancy on birthweight. The AOL approach instead learns an optimal treatment regime. This means to learn a decision rule mapping maternal characteristics to a smoking recommendation. %AOL learns this decision rule and estimates its value, that is, the expected birthweight under the learned policy. 
Given the overwhelming evidence of harm, we expect the estimated average treatment effects to be negative and of plausible magnitude, and the learned regime to universally recommend abstinence from smoking, with the corresponding optimal value reflecting the mean birthweight in a population following this recommendation. 

We emphasize that the purpose of this application is not to find new results, but rather to provide a credible validation of our methods against a well-understood ground truth. %At the same time, the framework allows us to examine heterogeneity in the magnitude of the birthweight penalty across maternal subpopulations, which has practical relevance for targeting smoking cessation interventions toward groups where the harm is most pronounced. 
Using this dataset, we estimate causal effects and construct confidence intervals via our proposed cBLB procedure, demonstrating its practical utility for scalable causal inference in large-scale observational studies, where the standard bootstrap is prohibitively expensive computationally and effectively infeasible at this scale.

\subsection{Related Work}
\label{sec:rel_work}

Kernel methods have been widely used in personalized medicine and optimal policy learning. Optimal policy learning seeks to develop decision rules that recommend treatments based on individual characteristics to maximize expected outcomes. This approach, also known as personalized medicine or individualized treatment rules (ITRs), acknowledges that individuals may respond differently to the same treatment, making personalization essential for optimal care \citep{zhao2012estimating}. Several methodological frameworks have emerged in the optimal policy learning literature. One is the transformation of treatment regime estimation into a weighted classification problem using kernel machines, as demonstrated by outcome weighted learning (OWL) \cite{zhao2012estimating}. OWL typically employs SVMs with kernel functions to flexibly model nonlinear treatment decision boundaries. To address some limitations of OWL-such as sensitivity to outcome shifts and its inability to handle negative outcomes---residual weighted learning (RWL) was proposed \cite{zhou2017residual}. RWL also relies on kernel-based classification but replaces outcome weights with residuals from regression models, improving robustness and finite-sample performance \cite{zhou2017residual}. Building on OWL and RWL, \cite{fu2019robust} developed Robust OWL (ROWL), which improves the stability of kernel-based treatment rules by using bounded, non-convex loss functions within the same learning framework. Another extension is augmented OWL (AOL) \citep{liu2018augmented}, which incorporates outcome regression through augmented inverse probability weighting while preserving the kernel-based classification structure, and which \citet{liu2018augmented} also extend to dynamic treatment regimes. \citet{zhou2017augmented} build on AOL with a counterfactual residual that pairs each arm's propensity score with the opposite arm's outcome regression, and give a kernelized formulation. This achieves semiparametric efficiency and avoids the non-convex optimization required by RWL.

Kernel methods naturally quantify distributional similarity through reproducing kernel Hilbert space (RKHS) embeddings \citep{hirshberg2021augmented, ben2021balancing, kallus2020generalized}, making them well-suited for covariate balancing methods that aim to estimate weights which balance covariate distributions across groups \citep{ben2021balancing}.
For example, \citet{kallus2020generalized} proposed a unifying framework that formulates matching and weighting as the problem of minimizing worst-case bias over a specified function class. They showed that several existing methods---including exact matching and propensity score weighting---can be interpreted as minimizing bias under different function classes, and introduced new kernel-based estimators that minimize worst-case bias in an RKHS. Subsequent work extended this framework to more complex settings, including marginal structural models, average treatment effects, and continuous exposures \citep{kallus2021more, kallus2022optimal, Santacatterina2021-SANOBO, kallus2019kernel}. In particular, Kernel Optimal Orthogonality Weighting (KOOW) \citep{ kallus2019kernel} minimizes the worst-case functional covariance between the treatment variable and confounders in an RKHS, while incorporating a regularization term to avoid extreme weights. Recent work has also addressed the computational bottleneck of kernel balancing in very large datasets by using the rank-restricted Nystr\"{o}m method to compute a kernel basis for balancing in nearly linear time and space, combined with scalable convex optimization for the weights \cite{kim2023scalable}. While this is an important contribution, it focuses on making a specific kernel weighting estimator scalable, whereas our proposed method makes bootstrap-based uncertainty quantification scalable for a broader class of kernel-based estimators.

Kernel methods also appear in the DML framework, a modern approach that uses machine learning for estimating nuisance functions while achieving unbiased, $\sqrt{n}$-consistent estimates of causal parameters \citep{chernozhukov2017double}. The idea behind DML is to estimate nuisance functions (e.g., the outcome $m(x)=E[Y|X=x]$ or propensity score $\pi(x)=P[W=1|X=x]$) with flexible methods, then form an orthogonalized estimating equation so that the first-order error in nuisance estimates does not bias the target parameter. Additionally, DML uses sample-splitting or cross-fitting to avoid regularization biases \citep{chernozhukov2017double, kennedy2023towards}. In this framework, any regression method achieving a sufficient rate of convergence can be used for the nuisances -- including kernel-based learners. For instance, \cite{naimi2017nonparametric} used a SuperLearner ensemble that included a SVM to estimate nuisance functions. Additionally, \cite{colangelo2020double} applied kernel-based double debiased machine learning estimators to estimate the effect of a continuous treatment.

When the bootstrap is used in causal inference, it typically involves drawing samples with replacement from the data and ``re-designing'' each sample to maintain covariate balance \citep{zhang2021designed,dagan2021bnt162b2}. There is more work on the bootstrap in the context of matching estimators; both \cite{abadie2008failure} and \cite{abadie2022robust} explored applying the bootstrap to matching. \cite{abadie2008failure} highlights that the standard bootstrap does not provide valid inference for matching estimators, requiring adjustments. Building on this, \citet{abadie2022robust} investigates how to compute valid standard errors for regression coefficients, including treatment effects, post-matching. Similarly, \cite{otsu2017bootstrap} addresses the issues raised by \citet{abadie2008failure} by developing a weighted bootstrap that does not count how many times an observation is used in estimating the bootstrap statistic, instead resampling observations directly. \cite{adusumilli2018bootstrap} introduces a modified bootstrap approach based on the concept of potential errors, while \cite{zhao2019sensitivity} applies the bootstrap to sensitivity analysis, providing confidence intervals for such analyses.

\textbf{Paper Organization.} The remainder of this paper is organized as follows. In Section~\ref{sec:setup}, we provide a detailed description of the kernel-based methods we consider in our simulation experiments and real-world application. In Section~\ref{sec:causalBLB}, we introduce our proposed method for fast uncertainty quantification for kernel-based estimators in large-scale causal inference and discuss some properties. In Section~\ref{sec:simulate}, we present the results of our Monte Carlo simulation studies evaluating the performance of our method across different settings and provide practical guidelines informed by these empirical results. In Section~\ref{sec:application}, we apply our method to the NVSS dataset to estimate the effect of smoking on birthweight and learn an optimal treatment regime. Finally, in Section~\ref{sec:conclusion}, we discuss the implications of our findings, limitations of our approach, and directions for future research.

\section{Setup: Kernel-based Estimators for Treatment Effect Estimation and Policy Learning}\label{sec:setup}

As described above, kernel methods are versatile tools used across a range of causal inference tasks. In this paper, we focus on three representative kernel-based methods that serve as case studies for our simulation experiments and real-world application. Specifically, we consider: (1) a kernelized version of the Augmented Outcome-Weighted Learning (AOL) method for estimating optimal treatment regimes \citep{zhou2017augmented}; (2) a covariate balancing estimator using kernel minimax weights within the Augmented Minimax Linear Estimation (AMLE) framework \citep{hirshberg2021augmented};  and (3) a DML estimator that employs kernelized SVM to learn nuisance functions \citep{kennedy2024semiparametric}.

Before describing the three methods, we recall the representer theorem and the RKHS facts used throughout. Let $\mathcal{H}_K$ denote the RKHS associated with a positive-definite
kernel $K(\cdot,\cdot)$, with inner product $\langle \cdot, \cdot \rangle_{\mathcal{H}_K}$, norm $\| \cdot \|_{\mathcal{H}_K}$, and reproducing property $h(x) = \langle h, K(\cdot,x) \rangle_{\mathcal{H}_K}$ for all $h \in \mathcal{H}_K$. The representer theorem \citep{kimeldorf1971some,scholkopf2001generalized} states that, for any loss $L$, the minimizer of a regularized empirical risk of the form $\frac{1}{n}\sum_{i=1}^n L(z_i, h(x_i)) + \lambda \|h\|_{\mathcal{H}_K}^2$ over $h \in \mathcal{H}_K$ admits the finite-dimensional representation $h(x) = \sum_{j=1}^n \alpha_j K(x,x_j)$ for some coefficients $\{\alpha_j\} \subset \mathbb{R}$, converting the (possibly infinite-dimensional) variational problem into a finite-dimensional optimization over $\{\alpha_j\}$. The kernelized AOL optimization in Section~\ref{sec:kAOL} relies on this. We also use, by the Cauchy--Schwarz inequality and the reproducing property, $\sup_{h \in \mathcal{H}_K, \|h\|_{\mathcal{H}_K} \leq 1} |\langle h, g \rangle_{\mathcal{H}_K}| = \|g\|_{\mathcal{H}_K}$ for any fixed $g \in \mathcal{H}_K$; this identity lets the worst-case imbalance $\mathbb{I}_{\mathcal{M}_a}(\gamma)$ in Section~\ref{sec:kernel_weights} be evaluated in closed form as a supremum over the RKHS unit ball $\mathcal{B}_{\mathcal{H}} = \{m_a \in \mathcal{H} : \|m_a\|_{\mathcal{H}} \leq 1\}$.

\subsection{Kernelized Augmented Outcome-Weighted Learning}
\label{sec:kAOL}

The goal of optimal treatment regime estimation is to learn a decision rule that maps patient covariates $X$ to a treatment assignment $A \in \{-1, +1\}$ in order to maximize expected clinical outcomes. AOL frames this problem as a weighted classification task, where the misclassification loss is weighted by counterfactual residuals derived from a doubly robust estimator. While linear decision functions are computationally attractive, they may fail to capture complex, nonlinear treatment boundaries in high-dimensional covariate spaces.

To address this, AOL can be extended to a nonlinear setting by modeling the decision function $f(x)$ in a reproducing kernel Hilbert space (RKHS). Specifically, the decision rule is parameterized as $f(x) = h(x) + b$, where $h \in \mathcal{H}_K$ and $\mathcal{H}_K$ is the RKHS associated with a positive-definite kernel $K(\cdot, \cdot)$. The learning objective is then to minimize a regularized empirical risk, with a convex surrogate loss applied to the sign-adjusted decision margin:
\[
\min_{h \in \mathcal{H}_K, \, b \in \mathbb{R}} \frac{1}{n} \sum_{i=1}^{n} \frac{|\tilde{r}_i|}{\pi(a_i, x_i)} \, \phi\left( a_i \cdot \text{sign}(\tilde{r}_i) \cdot (h(x_i) + b) \right) + \frac{\lambda}{2} \|h\|_{\mathcal{H}_K}^2,
\]
where $\phi(\cdot)$ is a convex surrogate loss (e.g., Huberized hinge loss), and $\tilde{r}_i$ denotes the counterfactual (doubly robust) residual for observation $i$,

\[
\tilde{r}_i = y_i - \hat{\tilde g}(x_i), \qquad \hat{\tilde g}(x_i) = \hat\pi(-1,x_i)\,\hat m(x_i,1) + \hat\pi(1,x_i)\,\hat m(x_i,-1),
\]
i.e., the AIPW pseudo-outcome contrast between the two treatment arms, derived in \cite{zhou2017augmented} (Eqs.~4 and~13); the convex reformulation of the loss using $|\tilde r_i|$ as the weight and $\mathrm{sign}(\tilde r_i)$ as the label follows \cite{liu2018augmented}.

By the representer theorem, the solution $h(\cdot)$ admits the finite representation $h(x) = \sum_{j=1}^{n} v_j K(x, x_j)$, converting the infinite-dimensional problem into a finite-dimensional convex optimization over coefficients $\{v_j\}$ and bias $b$. This formulation enables AOL to learn flexible, nonlinear treatment regimes while retaining global convergence guarantees and computational efficiency via standard solvers such as L-BFGS.

Despite its flexibility, kernelized AOL introduces computational challenges that scale with the number of training examples. Specifically, the kernel representation requires storing and manipulating the full $n \times n$ Gram matrix, leading to $\mathcal{O}(n^2)$ memory complexity in optimization scenarios. This issue can become prohibitive in large datasets. 

\subsection{Kernel Minimax Weights}
\label{sec:kernel_weights}

For both the covariate balancing estimator using kernel minimax weights and the DML approach employing support vector machines to estimate nuisance functions, we focus, for simplicity, on estimands such as the average treatment effect, $\mathbb{E}[Y(1) - Y(0)]$. Specifically, we consider $\mathbb{E}[Y(a)]$, which is nonparametrically identified under standard causal inference assumptions (ignorability, positivity, and consistency) as
\[
\mathbb{E}[Y(a)] = \mathbb{E} \left[ \mathbb{E}[Y \mid A = a, X] \right] = \E[ m_a(X) ] =: \Psi(m_a).
\]

We define $Y_i = m_a(X_i) + \epsilon_i$, which implies $\epsilon_i = Y_i - m_a(X_i)$ and $m_a(X_i) = Y_i - \epsilon_i$ where $\E[\epsilon_i \mid X_i] = 0$, and $\sigma^2 = \E[\epsilon_i^2 | A_i=a, X_i] = \textrm{Var}[Y_i | A_i=a, X_i]$ for $a=0,1$ and all $i$. Let $\delta_{m_a}(X_i) = \hat{m}_{a}(X_i) - m_{a}(X_i)$ (the regression error) and recall the augmented estimator for treatment $a$ we considered
\begin{align*}
    \hat{\Psi}^{a} &= \frac{ 1 }{n} \sum^{n}_{i=1}  \hat m_{a}(X_i)   \\ 
    &- \frac{ 1 }{n}\sum^{n}_{i=1} \big[ \gamma(A_i, X_i)(\hat m_{a}(X_i) - Y_i)\big], \\ 
\end{align*}
\noindent
where $\gamma(A_i, X_i) = \I{A_i=a} \gamma(X_i)$ and where $\gamma(X_i)$ could be chosen to be set to $\frac{1}{\pi(X_i)}$ as in section \ref{sec:DML_SVM} below. By looking at the error of this estimator we can decompose this estimator as 

\begin{align*}
    \hat{\Psi}^{a} - \Psi(m_a) &= \frac{ 1 }{n} \sum^{n}_{i=1} \big[ (\delta_{m_{a}}(X_i) + m_{a}(X_i) ) \big] \\ 
    &- \frac{ 1 }{n}\sum^{n}_{i=1}\big[ \gamma(A_i, X_i)\delta_{m_a}(X_i)   \big] \\ 
    &+ \frac{ 1 }{n}\sum^{n}_{i=1}\big[\gamma(A_i, X_i)\epsilon_i \big] - \Psi(m_a)\\
    &= \underbrace{ \frac{ 1 }{n} \Big[\sum^{n}_{i=1} \delta_{m_a}(X_i) - \sum^{n}_{i=1}\gamma(A_i,X_i)\delta_{m_a}(X_i) \Big]}_{\text{imbalance in } \delta_{m_a}}\\
    &+ \underbrace{ \frac{ 1 }{n} \sum^{n}_{i=1}\gamma(A_i,X_i)\epsilon_i}_{\text{noise - mean zero}} + \underbrace{ \frac{ 1 }{n} \sum^{n}_{i=1} m_{a}(X_i)  - \Psi(m_a)}_{\text{sampling variation - mean zero}}
\end{align*}

We assume that $\delta_{m_a}(X_i)$ is contained in an absolutely convex set of functions $\mathcal{M}$ (e.g., the unit ball of a Hilbert space). 

We then choose $\gamma(A_i,X_i)$ controlling the maximal imbalance in $\mathcal{M}$. In formulas we consider the following worst imbalance in $\mathcal{M}_a$,

\begin{center}
 $\mathbb{I}_{\mathcal{M}_a}(\gamma) = \sup_{\delta_{m_a} \in \mathcal{M}} \mid \frac{ 1 }{n} \Big(\sum^{n}_{i=1} \delta_{m_a}(X_i) - \sum^{n}_{i=1}\gamma(A_i,X_i)\delta_{m_a}(X_i) \Big)   \mid$   
\end{center}

Typically, $\gamma(A_i, X_i)$ is obtained by minimizing the worst-case imbalance while accounting for some measure of the complexity of the weights, i.e., a variance-weighted penalty of the form $\frac{\sigma^2}{n^2}\sum_i \gamma(A_i, X_i)^2$, where $\sigma^2$ is the outcome variance defined above ($\sigma^2 = \mathrm{Var}[Y_i \mid A_i = a, X_i]$).  Following standard practice \citep{hirshberg2021augmented,kallus2022optimal}, we choose as a model $\mathcal{M}_a$ the unit ball $\mathcal{B}_{\mathcal{H}} = \{ m_a(x) \in \mathcal{H} : \| m_a \|_{\mathcal{H}} \leq 1 \}$ of an RKHS $\mathcal{H}$.

Define the matrix $K_a\in\mathbb R^{n\times n}$ as $K_{aij}=\mathcal K_a(X_{i},X_{j})$ and setting $\gamma_i = \gamma(X_i)$ for convenience, so that $\gamma(A_i,X_i) = \I{A_i=a}\gamma_i$ as defined above. By the representer theorem, we have that

\begin{align*}
\mathbb{I}^2_{\mathcal{M}_a}(\gamma) &=  \sup_{\|m_a\|^2_{\mathcal{H}} \leq 1} \prns{ \frac{1}{n} \sum_{i = 1}^n  \prns{ 1-\I{A_i=a}\gamma_i} m_a(X_i)}^2 \\
&= \sup_{\sum_{i,j=1}^n\alpha_i\alpha_j\mathcal K_a(X_i,X_j)\leq1} \prns{ \frac{1}{n} \sum_{i = 1}^n  \prns{1 -\I{A_i=a}\gamma_i} \sum_{j=1}^n\alpha_j\mathcal K_a(X_i,X_j)}^2 \\
&= \sup_{\alpha^TK_a\alpha\leq1} \prns{ \frac{1}{n} \alpha^TK_a(e_{n}-I_{a}\gs)}^2 \\
&=\frac{1}{n^2}(I_{a}\gs-e_{n})^TK_a(I_{a}\gs-e_{n}) \\
&= \frac{1}{n^2} \prns{\gs^TI_{a}K_aI_{a}\gs-2 e_{n}^TK_aI_{a}\gs+e_{n}^T K_a e_{n}}
\end{align*}

\noindent
where $e_{n}$ is the length-$n$ vector with observations equal to 1, and $I_{a}$  is the $n$-by-$n$ diagonal matrix with $\I{A_i=a}$ in the $i^{th}$ diagonal entry. This leads to the following optimization problem
\begin{equation}
    \hat{\gamma} =  \underset{\gamma}{\arg\min}\Big[\mathbb{I}^{2}_{\mathcal{M}}(\gamma) + \lambda \frac{\sigma^2}{n^2}\sum_i  \gamma(A_i, X_i)^2\Big]
    \label{eq:kernel_minimax_weights}
\end{equation}
\noindent
where $\lambda$ is a penalization parameter, fixed here at $\lambda=\log(2)$, and 

\begin{align*}
    \mathbb{I}^{2}_{\mathcal{M}}(\gamma) = \mathbb{I}^{2}_{\mathcal{M}_1}(\gamma) + \mathbb{I}^{2}_{\mathcal{M}_0}(\gamma).
\end{align*}
Finally, the obtained weights $\hat{\gamma}$ are then plugged into the augmented estimator $\hat{\Psi}^{a}$. Following \cite[Section~2.6]{hirshberg2021augmented}, we cross-fit: the sample is partitioned into $K$ folds, stratified by treatment arm so that every fold contains treated and control units; for each fold, the outcome regression $\hat m_a$ is fit on the other $K-1$ folds, and the weights $\hat\gamma$ are then obtained by solving \eqref{eq:kernel_minimax_weights} on the held-out fold. As with DML (Section~\ref{sec:DML_SVM}), within cBLB the folds are formed inside each subset $I_k$, and we use $K=10$. Similar to AOL, this introduces computational challenges that scale with the number of training examples, making uncertainty calculations with standard bootstrap computationally prohibitive in large datasets. 

\subsection{Double Machine Learning with Support Vector Machines}
\label{sec:DML_SVM}

Estimating treatment effects via DML involves constructing an orthogonalized moment function (score) that enables robust, unbiased estimation of targets like the average treatment effect (ATE) or conditional ATE (CATE) while using flexible machine learning to learn nuisance functions \citep{chernozhukov2017double, kennedy2023towards}. The key property is Neyman orthogonality (or pathwise differentiability), which means the moment condition is insensitive to small errors in the nuisance estimates \citep{chernozhukov2017double}. 

This framework yields augmented estimators in the form of $\hat{\Psi}^{a}$ defined in Section \ref{sec:kernel_weights} with desirable properties \citep{chernozhukov2017double,rotnitzky2021characterization}. Implementation requires estimating the nuisance functions $(\hat \pi, \hat m_a)$ using flexible learners, and employs sample-splitting (cross-fitting) to avoid overfitting biases \citep{chernozhukov2017double,robins2008higher}. In cross-fitting, the sample is partitioned into $K$ folds; for each fold, the nuisance models are trained on the other $K-1$ folds and then used to compute the moment function on the held-out fold. Within cBLB, $K$-fold cross-fitting is performed independently within each subset $I_k$, after the full-sample partition into subsets: each fold therefore holds only $b/K$ observations rather than $n/K$. We use $K=10$ throughout. Partitioning first keeps nuisance fitting confined to each size-$b$ subset and avoids fitting on all $n$ observations. Algorithm~\ref{alg:dml_subset} in Section~\ref{sec:causalBLB} gives the exact subset-level pseudocode for this procedure.

In this work, we use kernel-based SVMs to learn the propensity score $\pi(X)$ and outcome regressions $m_a(X)$ within the DML procedure. SVMs are appealing as first-stage learners because they provide a flexible nonparametric fit while enforcing regularization via convex loss minimization. There are, however, important computational considerations when using kernel SVMs in a cross-fitted DML. In popular implementations of SVMs as in scikit-learn, an SVM's time complexity scales between $\mathcal{O}(n_\textrm{features} \times n^2)$ and $\mathcal{O}(n_\textrm{features} \times n^3)$, making uncertainty calculations with standard bootstrap computationally infeasible for large datasets. \citep{scikit-learn, chapelle2007training}.

\section{Causal Bag of Little Bootstraps}
\label{sec:causalBLB}
We observe i.i.d.\ data $Z_i=(Y_i,A_i,X_i)$, $i=1,\dots,n$, where $A_i\in\{0,1\}$ is treatment, $X_i\in\mathcal X$ are covariates, and $Y_i=Y_i(A_i)$ under the potential outcomes framework \citep{imbens2015causal}. Causal Bag of Little Bootstrap (Algorithm~\ref{alg:cblb}) partitions the indices into $s$ disjoint subsets $I_k$ of size $b$ (chosen to cover almost all observations). Within each subset $I_k$, we fit all required objects (nuisance functions, kernel weights, and/or a decision rule) once using $\{Z_i:i\in I_k\}$, compute per-unit contributions $\{\hat\theta_{i,k}: i\in I_k\}$, and then generate $r$ bootstrap replicates by multinomial reweighting of these fixed contributions; no refitting is performed within replicates.

The definition of $\hat\theta_{i,k}$ depends on the target. In this paper we consider inference for kernelized AOL policy learning and for the ATE:
\begin{equation}
\hat\theta_{i,k}
\;=\;
\begin{cases}
\hat\theta^{\mathrm{kAOL}}_{i,k}
\;:=\;
\dfrac{Y_i}{\hat\pi_k(A_i \mid X_i)}\,\mathbf{1}\big(A_i = \hat d_k(X_i)\big),
& \text{(Section~\ref{sec:kAOL}), with } \hat d_k(\cdot) = \mathrm{sign}\big(\hat h_k(\cdot)+\hat b_k\big) \\[1.2em]
\hat\theta^{\mathrm{ATE}}_{i,k}
\;:=\;
\hat\psi_{i,k}(1)-\hat\psi_{i,k}(0),
& \text{(Sections~\ref{sec:kernel_weights}--\ref{sec:DML_SVM}).}
\end{cases}
\label{eq:theta_ik}
\end{equation}
Here, for $a\in\{0,1\}$,
\begin{equation}
\hat\psi_{i,k}(a)
\;:=\;
\hat m_{a,k}(X_i)\;-\;\mathbf{1}(A_i=a)\,\hat\gamma_{a,k}(X_i)\,\big\{\hat m_{a,k}(X_i)-Y_i\big\},
\label{eq:psi_ik}
\end{equation}
with $(\hat m_{a,k},\hat\pi_k,\hat\gamma_{a,k},\hat h_k,\hat b_k)$ fit on $I_k$. In Section~\ref{sec:kernel_weights} (kernel minimax weights), $\hat\gamma_{a,k}(\cdot)$ denotes the kernel-minimax weights estimated within $I_k$ via \eqref{eq:kernel_minimax_weights}, with $K$-fold cross-fitting inside $I_k$ (Algorithm~\ref{alg:minimax_subset}). In Section~\ref{sec:DML_SVM} (DML), $\hat\gamma_{a,k}(X_i)=1/\hat\pi_{a,k}(X_i)$, where $\hat\pi_{a,k}$ is the subset-trained propensity score.

This follows the same computational principle emphasized in the DML bootstrap literature: resample score/contribution terms to capture first-order uncertainty while avoiding repeated refitting of nuisance models. For example, \cite{belloni2018uniformly} studies a multiplier bootstrap for orthogonal scores, and this approach underlies the confidence intervals implemented in the \texttt{DoubleML} package \citep{bach2022doubleml} for \cite{chernozhukov2017double}. While this is computationally-efficient, \cite{tang2024consistency} notes that not refitting the nuisance functions can lead to poor finite-sample coverage and invalid confidence intervals when certain assumptions are not satisfied. Nevertheless, we find in our simulations that nominal coverage is retained even when not refitting in each resample. This is corroborated by Theorem~\ref{thm:norefit}, which shows that, once the per-unit contributions have been computed on a subset, multinomial reweighting of these fixed contributions reproduces the correct large-sample sampling variability of the subset estimator, so refitting within bootstrap replicates is unnecessary for first-order coverage. Algorithm~\ref{alg:cblb} constructs percentile confidence intervals (CIs), denoted $\textrm{CI}(\theta)$ below, from the bootstrap replicates within each subset and then averages across subsets.

\begin{algorithm}
\begin{algorithmic}
\Require Number of subsets $s$, their size $b$, and number of bootstrap replicates, $r$
\State Randomly permute $\{1, \ldots, n\}$ and partition into disjoint blocks $I_1, \ldots, I_s$ of size $b$, writing $I_k = \{i_1, \ldots, i_b\}$
\For{$k \gets 1$ to $s$}
    \State Fit any required nuisance models and/or policy model on the subset $\{Z_i : i \in I_k\}$, and compute per-index contributions
    $\hat{\theta}_{i_j}$ (one value per $i_j \in I_k$), where $\hat{\theta}_{i_j}$ may depend on the fitted objects on subset $k$
    \For{$\ell \gets 1$ to $r$}
        \State Sample $\left(M_{1}, \ldots, M_{b}\right) \sim \textrm{Multinomial}\left(n; \frac{1}{b}, \ldots, \frac{1}{b}\right)$
        \State Calculate the bootstrap empirical distribution $\mathbf{P}^\ast_{n, \ell} \gets \frac{1}{n}\sum_{a=1}^b M_a\delta_{Z_{i_a}}$
        \State Calculate the estimator as the reweighted average of the contributions:
        $\hat{\theta}_{n, \ell} \gets \frac{1}{n}\sum_{a=1}^b M_a \hat{\theta}_{i_a}$
    \EndFor
    \State $\textrm{CI}(\theta)_k \gets (\hat{q}_{k,0.025}, \hat{q}_{k,0.975})$, where $\hat{q}_{k,\alpha}$ denotes the empirical $\alpha$-quantile of $\bigl\{\hat{\theta}_{n,\ell}\bigr\}_{\ell=1}^{r}$
\EndFor
\State $\textrm{CI}(\theta) \gets \left(\frac{1}{s}\sum_{k=1}^s\hat{q}_{k,0.025}, \frac{1}{s}\sum_{k=1}^s\hat{q}_{k,0.975}\right)$
\end{algorithmic}
\caption{cBLB}\label{alg:cblb}
\end{algorithm}

Algorithms~\ref{alg:kaol_subset}, \ref{alg:minimax_subset}, and \ref{alg:dml_subset} below detail the subset-level fitting step of Algorithm~\ref{alg:cblb} (line~3) for each of the three estimators of Section~\ref{sec:setup}, and show explicitly how the per-unit contribution $\hat\theta_{i,k}$ of \eqref{eq:theta_ik} is computed in each case.

\begin{algorithm}
\begin{algorithmic}
\Require Subset $I_k=\{i_1,\ldots,i_b\}$
\State Estimate the propensity $\hat\pi_k(A_i\mid X_i)$ on $\{Z_i:i\in I_k\}$ 
\State Fit the outcome regression $\hat m_k(X_i,a)$, $a\in\{-1,1\}$, on $\{Z_i:i\in I_k\}$ and form the counterfactual residual
\Statex \hspace{1.5em} $\tilde r_{i,k} \;\gets\; Y_i - \big\{\hat\pi_k(-1\mid X_i)\,\hat m_k(X_i,1) + \hat\pi_k(1\mid X_i)\,\hat m_k(X_i,-1)\big\}$ \Comment{as in the residual formula of Section~\ref{sec:kAOL}}
\State \textbf{Fit the decision rule} by solving the kernelized AOL problem on $I_k$ for $(\hat v_k,\hat b_k)$, using the Gram matrix $K(X_i,X_j)$, $i,j\in I_k$:
\Statex \hspace{1.5em} $(\hat v_k,\hat b_k) \;\gets\; \displaystyle\argmin_{v,\,\beta}\; \frac{1}{|I_k|}\sum_{i\in I_k} \frac{|\tilde r_{i,k}|}{\hat\pi_k(A_i\mid X_i)}\, \phi\!\big(A_i\,\mathrm{sign}(\tilde r_{i,k})\,(K_i^\top v + \beta)\big) \;+\; \frac{\lambda}{2}\,v^\top K v$
\State Set the estimated decision rule $\hat d_k(X_i) \;\gets\; \mathrm{sign}\big(K_i^\top \hat v_k + \hat b_k\big)$ for $i \in I_k$, where $K_i$ denotes the $i$-th row of the subset Gram matrix
\For{$i \in I_k$}
    \State $\hat\theta^{\mathrm{kAOL}}_{i,k} \;\gets\; \dfrac{Y_i}{\hat\pi_k(A_i \mid X_i)}\,\mathbf{1}\big(A_i = \hat d_k(X_i)\big)$ \Comment{cf.\ \eqref{eq:theta_ik}}
\EndFor
\State \Return $\{\hat\theta^{\mathrm{kAOL}}_{i,k}: i \in I_k\}$
\end{algorithmic}
\caption{Subset-level fitting: kernelized AOL (Section~\ref{sec:kAOL})}\label{alg:kaol_subset}
\end{algorithm}

\begin{algorithm}
\begin{algorithmic}
\Require Subset $I_k=\{i_1,\ldots,i_b\}$, number of cross-fitting folds $K$
\State Randomly partition $I_k$ into $K$ folds, stratified by treatment arm
\For{each fold $f$}
    \For{$a \in \{0,1\}$}
        \State Fit outcome regression $\hat m_{a,k}(\cdot)$ on $\{Z_i:i\in I_k\setminus f\}$ and predict $\hat m_{a,k}(X_i)$ for $i\in f$
        \State Form the kernel matrix $K_{a,f}$, $[K_{a,f}]_{ij}=\mathcal K_a(X_i,X_j)$, $i,j\in f$, with hyperparameters from the fit on $I_k\setminus f$
        \State Solve for $\hat\gamma_{a,k}(X_i)$, $i\in f$, by \eqref{eq:kernel_minimax_weights} on $f$
    \EndFor
\EndFor
\For{$i \in I_k$}
    \State $\hat\psi_{i,k}(a) \gets \hat m_{a,k}(X_i) - \mathbf 1(A_i=a)\,\hat\gamma_{a,k}(X_i)\{\hat m_{a,k}(X_i)-Y_i\}$ for $a=0,1$ \Comment{\eqref{eq:psi_ik}}
    \State $\hat\theta^{\mathrm{ATE}}_{i,k} \gets \hat\psi_{i,k}(1)-\hat\psi_{i,k}(0)$ \Comment{\eqref{eq:theta_ik}}
\EndFor
\State \Return $\{\hat\theta^{\mathrm{ATE}}_{i,k}: i \in I_k\}$
\end{algorithmic}
\caption{Subset-level fitting: kernel minimax weights (Section~\ref{sec:kernel_weights})}\label{alg:minimax_subset}
\end{algorithm}

\begin{algorithm}
\begin{algorithmic}
\Require Subset $I_k=\{i_1,\ldots,i_b\}$, number of cross-fitting folds $K$
\State Randomly partition $I_k$ into $K$ folds of size $b/K$
\For{fold $f = 1,\ldots,K$}
    \State Fit kernel-SVM propensity $\hat\pi_{a,k}$ and outcome regressions $\hat m_{a,k}$ ($a=0,1$) on $I_k\setminus f$
    \State Predict $\hat\pi_{a,k}(X_i)$, $\hat m_{a,k}(X_i)$ for $i \in f$ using the models trained on $I_k\setminus f$
\EndFor
\State Set $\hat\gamma_{a,k}(X_i) \gets 1/\hat\pi_{a,k}(X_i)$ for $a=0,1$, $i \in I_k$
\For{$i \in I_k$}
    \State $\hat\psi_{i,k}(a) \gets \hat m_{a,k}(X_i) - \mathbf 1(A_i=a)\,\hat\gamma_{a,k}(X_i)\{\hat m_{a,k}(X_i)-Y_i\}$ for $a=0,1$ \Comment{\eqref{eq:psi_ik}}
    \State $\hat\theta^{\mathrm{ATE}}_{i,k} \gets \hat\psi_{i,k}(1)-\hat\psi_{i,k}(0)$ \Comment{\eqref{eq:theta_ik}}
\EndFor
\State \Return $\{\hat\theta^{\mathrm{ATE}}_{i,k}: i \in I_k\}$
\end{algorithmic}
\caption{Subset-level fitting: DML with kernel SVM (Section~\ref{sec:DML_SVM})}\label{alg:dml_subset}
\end{algorithm}

\subsection{Properties}
\label{subsec:properties}

As described above, cBLB executes all expensive fitting steps (nuisance estimation, kernel-weight optimization, and/or policy learning) once per subset $I_k$ and then holds the fitted objects fixed across the $r$ bootstrap replicates within that subset. Our simulations show that this ``no-refit within replicates'' implementation attains nominal coverage in finite samples in our main simulation settings (Section~\ref{sec:simulate}). Here we justify this choice theoretically by showing that, under standard asymptotic linearity of the corresponding full-sample estimator and mild stability of the subset-level contributions, cBLB still attains first-order (asymptotically correct) coverage despite refitting only at the subset level.

\begin{theorem}[First-order validity of cBLB (no refit)]
\label{thm:norefit}
Fix a subset $I_k=\{i_1,\ldots,i_b\}$ of size $b$ and fitted objects computed on $\{Z_i:i\in I_k\}$, yielding contributions
$\{\hat\theta_{i,k}: i\in I_k\}$ and the subset estimator
\[
\hat\theta_{k} := \frac{1}{b}\sum_{i\in I_k}\hat\theta_{i,k}.
\]
Let $M=(M_1,\ldots,M_b)\sim\mathrm{Multinomial}(n;1/b,\ldots,1/b)$ and define the cBLB replicate
\[
\hat\theta^*_{n,k}:=\frac{1}{n}\sum_{a=1}^b M_a\,\hat\theta_{i_a,k}.
\]
Assume:
(i) the corresponding full-sample estimator admits an influence-function expansion with influence function $\psi\in L_2(P)$;
(ii) within the fixed subset $I_k$, the frozen contributions satisfy the first-order approximation 
\begin{equation}
\hat\theta_{i,k}=\theta_0+\psi(Z_i)+r_{i,k},
\qquad
\frac{1}{\sqrt{b}}\sum_{i\in I_k} r_{i,k}=o_p(1), \mbox{ and }\frac{1}{b}\sum_{i\in I_k} r_{i,k}^2=o_p(1)
\label{eq:influence_expansion}
\end{equation}
(iii) $b\to\infty$, $b/n\to 0$, and the partition covers all but $o(n)$ observations. 

Then, conditional on $\{Z_i:i\in I_k\}$,
\[
\sqrt{n}\big(\hat\theta^*_{n,k}-\hat\theta_k\big)\ \Rightarrow\ N(0,\sigma_k^2),
\qquad 
\sigma_k^2 := \mathrm{Var}_{P_{b,k}}\!\big(\hat\theta_{i,k}\big),
\]
where $P_{b,k}$ denotes the empirical distribution on $\{Z_i:i\in I_k\}$.
%Moreover, under (ii) and $b\to\infty$, $\frac{1}{s} \sum_{k=1}^s \sigma_k^2  \xrightarrow{p} \operatorname{Var}\{\psi(Z)\}$ (or similarly $\forall \varepsilon > 0,  \frac{1}{s} \sum_{k=1}^s  \mathbf{1}\!\left\{ \left| \sigma_k^2 - \operatorname{Var}\{\psi(Z)\} \right| > \varepsilon \right\} \xrightarrow{p} 0$). 
Moreover, under (ii) and $b\to\infty$, for every $\varepsilon>0$,
\[
\frac{1}{s} \sum_{k=1}^s 
\mathbf{1}\!\left\{ 
\left| \sigma_k^2 - \operatorname{Var}\{\psi(Z)\} \right| 
> \varepsilon 
\right\}
\xrightarrow{p} 0.
\]

Finally, if in addition the aggregation conditions of Remark~\ref{rem:aggregation} hold and the number of bootstrap replicates per subset $r=r_n\to\infty$, the cBLB confidence intervals produced by Algorithm~\ref{alg:cblb} achieve asymptotically correct coverage for $\theta_0$.
\end{theorem}

\textit{Proof sketch:} Theorem~\ref{thm:norefit} implies that first-order coverage is governed by the estimator's leading $\sqrt{n}$ fluctuation. Within a subset $I_k$, holding fitted objects fixed turns $\{\hat\theta_{i,k}:i\in I_k\}$ into fixed per-unit contributions and $\hat\theta_k=b^{-1}\sum_{i\in I_k}\hat\theta_{i,k}$ into their empirical mean. A cBLB replicate $\hat\theta^*_{n,k}$ is the mean of $n$ draws from the empirical distribution $P_{b,k}$ on these $b$ contributions, so conditional on $\{Z_i:i\in I_k\}$ a conditional CLT yields $\sqrt{n}(\hat\theta^*_{n,k}-\hat\theta_k)\Rightarrow N(0,\sigma_k^2)$ with $\sigma_k^2=\mathrm{Var}_{P_{b,k}}(\hat\theta_{i,k})$. 
Assumption (ii), the influence-function expansion \eqref{eq:influence_expansion}, implies $\sigma_k^2 \to_p \mathrm{Var}\{\psi(Z)\}$  for a fixed subset $I_k$ as $b \to \infty$.  Since the equal-size subsets are identically distributed and
$\sigma_k^2 \to_p \operatorname{Var}\{\psi(Z)\}$ for each fixed $k$,
the proportion of subsets for which $\sigma_k^2$ differs from
$\operatorname{Var}\{\psi(Z)\}$ by more than any fixed
$\varepsilon>0$ converges to zero in probability.
Under the aggregation conditions of Remark~\ref{rem:aggregation}, averaging across subsets then transfers this first-order validity to the full-sample target, yielding asymptotically correct cBLB confidence intervals for $\theta_0$. A complete proof is given in Appendix~\ref{app:proof}.

Also note that Algorithm~\ref{alg:cblb} uses percentile (quantile) intervals from the uncentered replicates $\{\hat\theta^*_{n,k,j}\}_{j=1}^r$, and Theorem~\ref{thm:norefit} still applies because conditional quantiles shift by $\hat\theta_k$ (i.e., $Q_\alpha(\hat\theta^*_{n,k}\mid I_k)=\hat\theta_k+Q_\alpha(\hat\theta^*_{n,k}-\hat\theta_k\mid I_k)$), so first-order validity for $\hat\theta^*_{n,k}-\hat\theta_k$ implies first-order validity of the percentile CI after BLB aggregation.

%For the ATE target, Theorem~\ref{thm:norefit} applies because $\hat\theta^{\mathrm{ATE}}_{i,k}=\hat\psi_{i,k}(1)-\hat\psi_{i,k}(0)$, as defined in \eqref{eq:theta_ik}--\eqref{eq:psi_ik}, is an orthogonal-score/AIPW contribution computed using objects fit on $I_k$. Existing results for DML and augmented minimax linear / kernel-minimax-weight estimators establish the required influence-function expansion under standard regularity and nuisance-rate conditions
For the ATE target, assumption~(i) follows from \cite[Theorem~2 and Section~2.6]{hirshberg2021augmented}  for augmented minimax/kernel-minimax estimators, and from \cite[Assumption~5.1 and Theorem~5.1]{chernozhukov2018double} and \cite[Propositions~1--2]{kennedy2024semiparametric} for cross-fitted DML. Applied within each subset, these results give the expansion \eqref{eq:influence_expansion} under their nuisance-rate conditions at subset size $b$, so, together with the aggregation conditions of Remark~\ref{rem:aggregation}, Theorem~\ref{thm:norefit} applies to both ATE estimators.

For kernelized AOL, the conditional multinomial component of
Theorem~\ref{thm:norefit} applies to the subset-learned scalar criterion:
within each subset $I_k$ we learn $(\hat h_k,\hat b_k)$ (and any nuisance
components entering $\tilde r_{i,k}$ and $\hat\pi_k$) once, then treat
the resulting contributions as frozen and bootstrap
\[
\hat\theta_k^{\mathrm{kAOL}}
=
\frac{1}{b}\sum_{i\in I_k}
\hat\theta^{\mathrm{kAOL}}_{i,k}
\]
by multinomial reweighting of
$\{\hat\theta^{\mathrm{kAOL}}_{i,k}:i\in I_k\}$.
Conditional on the realized frozen contributions
$\{\hat\theta^{\mathrm{kAOL}}_{i,k}:i\in I_k\}$, the cBLB replicate is
the mean of $n$ i.i.d.\ draws from their empirical distribution.
Therefore, the conditional multinomial approximation in
Theorem~\ref{thm:norefit} applies around the realized subset estimator
$\hat\theta_k^{\mathrm{kAOL}}$. 
Because the final interval averages the subset-level percentile endpoints, its width reflects sampling variability of the value at the fixed rules $\hat d_k$, while variability of the learned rules across subsets shifts only the subset estimates $\hat\theta_k^{\mathrm{kAOL}}$.

The target is the optimal value $\theta_0=V(d^\ast)$. Theorem~\ref{thm:norefit}
applies to kernelized AOL if assumption~(ii) holds, which requires the learned
rule $\hat d_k$ to converge to $d^\ast$ quickly enough that the remaining error
in \eqref{eq:influence_expansion} is negligible. Establishing this rate for
kernelized AOL is beyond the scope of this paper, so for this estimator the
theorem provides a heuristic justification; the coverage results in
Table~\ref{tab:sim_coverage} show that cBLB performs well in practice.

\subsubsection{Computational complexity}
Let $C_{\mathrm{fit}}(m)$ be the cost of running the full fitting pipeline (nuisance fits and any kernel-weight/policy optimization) on $m$ observations.
A standard bootstrap with $r$ refits costs
\[
T_{\mathrm{boot}}=\Theta\!\big(r\,C_{\mathrm{fit}}(n)\big).
\]
cBLB fits once per subset and only reweights within replicates, so
\[
T_{\mathrm{cBLB}}
=\Theta\!\big(s\,C_{\mathrm{fit}}(b)\big)+\Theta(srb).
\]
With disjoint subsets covering most observations ($s\asymp n/b$),
\[
T_{\mathrm{cBLB}}
=\Theta\!\Big(\frac{n}{b}\,C_{\mathrm{fit}}(b)\Big)+\Theta(nr).
\]
Thus, when $C_{\mathrm{fit}}(m)$ grows faster than linearly in $m$ (typical for kernel methods/convex programs), taking $b=n^\gamma$ with $\gamma\in(0,1)$ yields a strict asymptotic improvement over $T_{\mathrm{boot}}$.

\section{Simulation}
\label{sec:simulate}

We conduct three simulations considering three estimators: AOL for policy learning, and kernel minimax weights, and DML for estimation of causal effects.

{\bf Data-generating mechanism}: For policy learning, the data-generating process (DGP) comes from \cite{zhou2017augmented} and is shown in \eqref{eq:poldgp}.

\begin{align}\label{eq:poldgp}
X_1, X_2, \ldots X_5 &\overset{\textrm{i.i.d.}}{\sim} \textrm{Unif}(-1, 1) \nonumber \\ 
\textrm{Pr}(W = 1 | \mathbf{X}) &= \textrm{Pr}(W = -1 | \mathbf{X}) = 0.5 \nonumber \\
\epsilon &\sim \textrm{Normal}(0, 1) \nonumber \\
Y &= 0.5 + 0.5X_1 + 0.8X_2 + 0.3X_3 - 0.5X_4 + 0.7X_5 + \\
&+ W(0.2 - 0.6X_1 - 0.8X_2) + \epsilon \nonumber
\end{align}

In this case, the optimal policy is given by:
\[
\delta(\mathbf{X}) = \begin{cases}
    1 & 0.2 - 0.6X_1 - 0.8X_2 > 0 \\
    -1 & \textrm{otherwise}
\end{cases}
\]

For estimation of causal effects with kernel minimax weights and DML, the DGP is delineated in equation \eqref{eq:dgm}.
\begin{align}\label{eq:dgm}
X_1, X_2 &\overset{\textrm{i.i.d.}}{\sim} \textrm{Normal}(0, 1) \nonumber \\ 
\textrm{Pr}(W = 1 | X_1, X_2) &= \dfrac{1}{1 + \exp(- 0.5X_1 - 0.5X_2)} \nonumber \\
\epsilon &\sim \textrm{Normal}(0, 1) \\
Y(0) &= X_1 + X_2 + \epsilon \nonumber \\
Y(1) &= Y(0) + \tau W \nonumber.
\end{align}
For each individual subject $i$, we generate two independent confounding variables $(X_{1i}, X_{2i})$ and calculate the propensity score by applying the inverse-logit function to a linear amalgamation of these covariates. An error term $\epsilon$ is then drawn independently and identically from a standard normal distribution to construct the outcomes under control and treatment conditions. The treatment effect remains constant across all subjects, with $\tau = 0.8$. The level of overlap between covariates is regulated by a specific parameter in the inverse-logit equation; a value of $-\frac{1}{2}$ is used to create a certain degree of imbalance between the treatment and control groups in relation to some of the covariates. 

{\bf Estimand}: For the optimal policy case, the estimand is the optimal value, that is, the expected outcome when the optimal policy is followed. For DML and kernel minimax weights, the estimand of interest $\tau$ is the average treatment effect (ATE): $\tau = E(Y(1) - Y(0))$. 

{\bf Methods}: We construct \reps{} replications of the DGP varying the number of subsets $s$. 

For  policy learning, we estimate individualized treatment regimes by solving a regularized empirical risk minimization problem using an augmented outcome-weighted learning (AOL) loss function:

\[
\min_f \dfrac{1}{n} \sum_{i=1}^n \dfrac{|\tilde r_i|}{\pi(a_i, \mathbf{x}_i)}\phi\left(a_i \operatorname{sign}(\tilde r_i) f(\mathbf{x}_i)\right) + \dfrac{\lambda}{2}||f||^2
\]

where $f$ is the decision boundary (represented by $h(\mathbf{x}) + b$ with $h(\mathbf{x}) \in \mathcal{H}_K$, a reproducing kernel Hilbert space associated with the linear
kernel function $K$), $\tilde r_i$ is a residual of a regression of the outcome on a function of covariates, $||f||$ is the norm $||\cdot||_K$, $\phi$ is the Huberized hinge loss function, and $\lambda$ is a tuning parameter. In the simulations, $\pi(a_i,\mathbf{x}_i)$ is set to the known randomization probability $0.5$, and the objective is minimized by L-BFGS-B with a finite-difference gradient. 

For DML, both the outcome and propensity score models are fit using SVM with a linear kernel. For kernel minimax weights, as proposed previously \citep{kallus2021more,kallus2022optimal}, we used a polynomial kernel $\mathcal{K}_a(z,z') = C\,(z^\top z')^{d_1} + \sigma_a^2\,\delta^\ast(z,z')$, fixing $d_1=1$ and estimating $C$ and $\sigma_a^2$ as hyperparameters on the training folds of 10-fold cross-fitting (Section~\ref{sec:kernel_weights}); for data generation/modeling we assumed the arm-specific regression functions $m_1(X)$ and $m_0(X)$ follow Gaussian process priors with kernels $\mathcal K_1$ and $\mathcal K_0$, and outcomes satisfy $Y_i = m(Z_i) + \varepsilon_i$ with $\varepsilon_i \sim \mathcal N(0,\sigma_a^2)$.

{\bf Performance Measures}:
We construct percentile confidence intervals \citep{efron1994introduction}. We want to ensure that the bootstrapped confidence intervals have at least nominal coverage, i.e.,

\[
\textrm{Pr}(\tau \in \textrm{CI}^\ast) \geq 1 - \alpha.
\]

{\bf Full Sample Timing Comparison}: To show the timing advantages of the cBLB, we compare the timing of the cBLB to a fair comparison of the full-sample bootstrap, where we do not re-fit or re-optimize in each bootstrap resample. All optimization/fitting is done once. Throughout, the full-sample bootstrap is cBLB with $s=1$ and $b=n$: the estimator is fitted once on all $n$ observations and the $r$ replicates reweight its per-unit contributions without refitting.

\subsection{Results}

\textbf{Coverage}  Figure~\ref{fig:zip_combined} shows ``zip plots'' of the confidence intervals for varying $s$ \citep{morris2019using}, one panel per estimator and subset count, and Table~\ref{tab:sim_coverage} reports the corresponding coverage values. We see that we are able to obtain nominal coverage in each case (although DML is slightly conservative).

\begin{table}[t]
\centering
\caption{\label{tab:sim_coverage}Coverage of cBLB confidence intervals across the three estimators ($n=5{,}000$, 1{,}000 replications)}
\begin{tabular}{lll}
\toprule
Estimator & Subsets ($s$) & Coverage\\
\midrule
Kernel minimax weights & 5 & 0.951\\
Kernel minimax weights & 10 & 0.949\\
Kernel minimax weights & 15 & 0.944\\
DML with kernel SVM & 5 & 0.963\\
DML with kernel SVM & 10 & 0.972\\
DML with kernel SVM & 15 & 0.964\\
Kernelized AOL & 5 & 0.962\\
Kernelized AOL & 10 & 0.956\\
Kernelized AOL & 15 & 0.948\\
\bottomrule
\end{tabular}
\end{table}

\begin{figure}
    \centering
    \includegraphics[scale = 0.6]{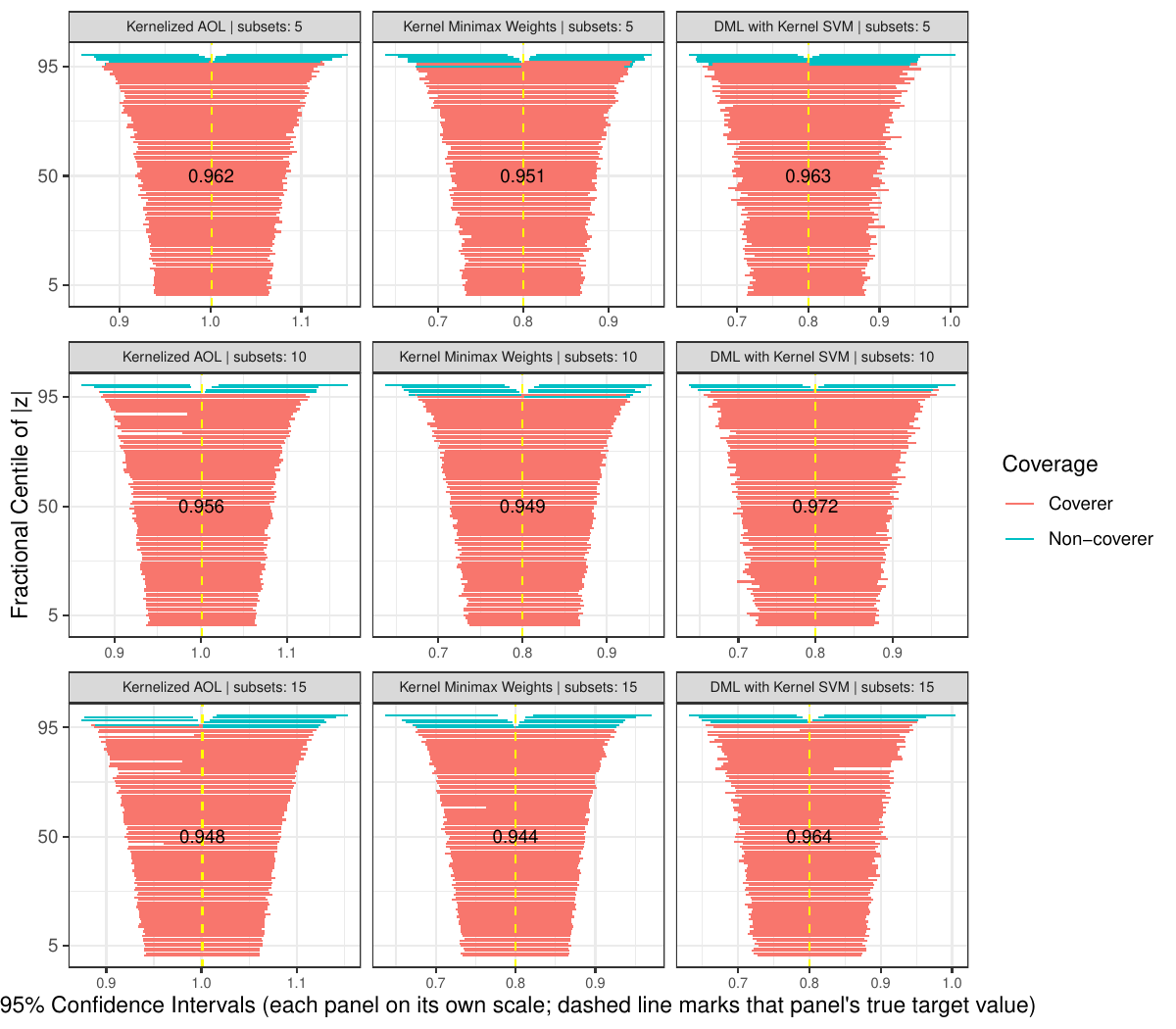}
    \caption{Confidence intervals for the estimand of each estimator from \reps{} replications of the cBLB algorithm ($n=5{,}000$): Kernelized AOL (optimal value), Kernel Minimax Weights (ATE), and DML with kernel SVM (ATE). Each row is plotted on its own $x$-axis scale, since the AOL row targets the optimal policy value rather than the ATE.}
\label{fig:zip_combined}
\end{figure}

\noindent\textbf{Timing.} Figure~\ref{fig:timing_combined} shows the comparison of the full-sample bootstrap to our cBLB for \timereps{} replications, fewer than the \reps{} replications used for coverage above: execution time varies far less across runs than statistical estimates do, and the full-sample bootstrap comparisons are computationally expensive. \footnote{The full-sample bootstrap here mimics cBLB in not refitting within each bootstrap resample.} Each panel is a standard box-and-whisker plot of execution time across the \timereps{} replications for a given method: the box spans the interquartile range, the line inside it marks the median, and the whiskers extend to the most extreme values within 1.5 times the interquartile range. All timing experiments were run on a single machine with an Apple M2 CPU (8 cores) and 8~GB of memory. For kernel minimax weights, both cBLB and the full-sample bootstrap use 10-fold cross-fitting, and the Gaussian process fits used single-threaded linear algebra (\texttt{OPENBLAS\_NUM\_THREADS=1}) in both.  In addition, Figure \ref{fig:timing_scaling_policy} shows how poorly the full-sample bootstrap scales for large data in the policy learning case.

\begin{figure}
    \centering
    \includegraphics[scale = 0.75]{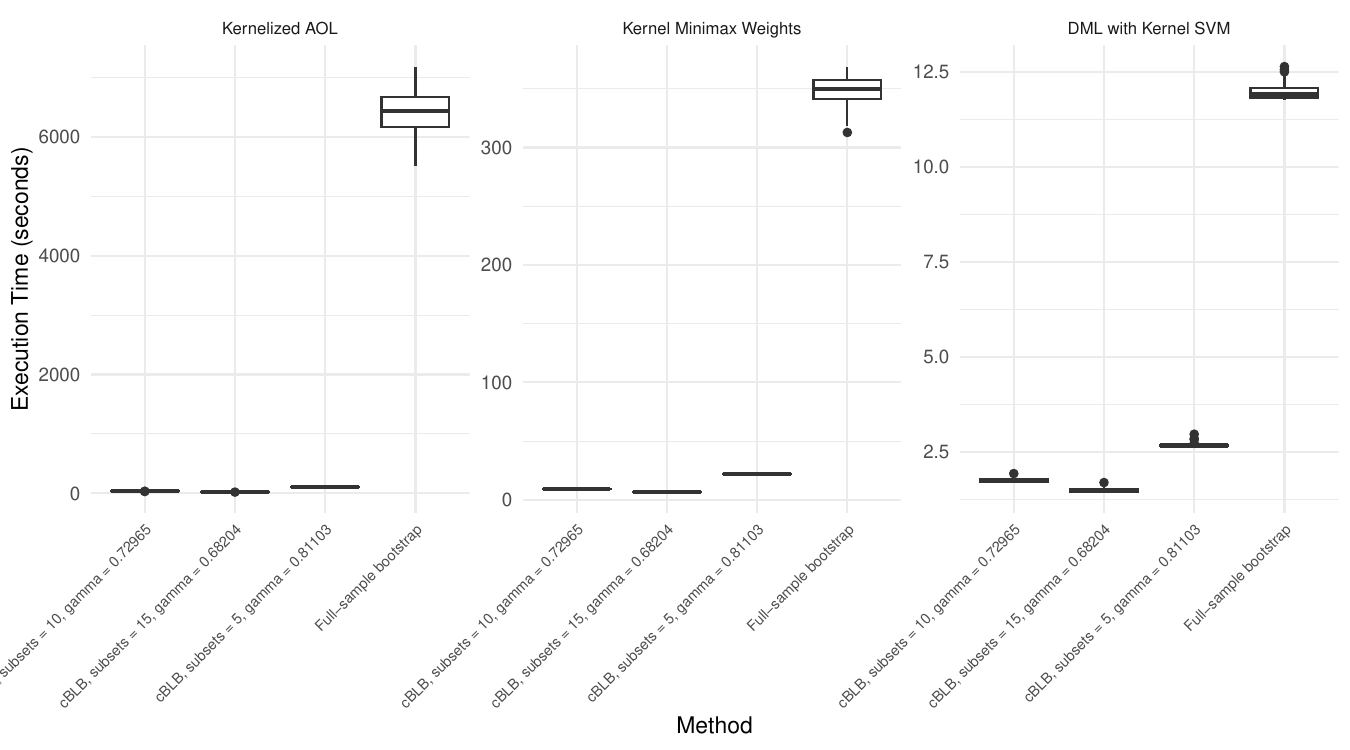}
    \caption{Timing results from \timereps{} replications ($n = 5000$) comparing cBLB to the full-sample bootstrap, one panel per estimator; note the differing $y$-axis scales across panels.}
    \label{fig:timing_combined}
\end{figure}

\begin{figure}
    \centering
    \includegraphics[scale = 0.75]{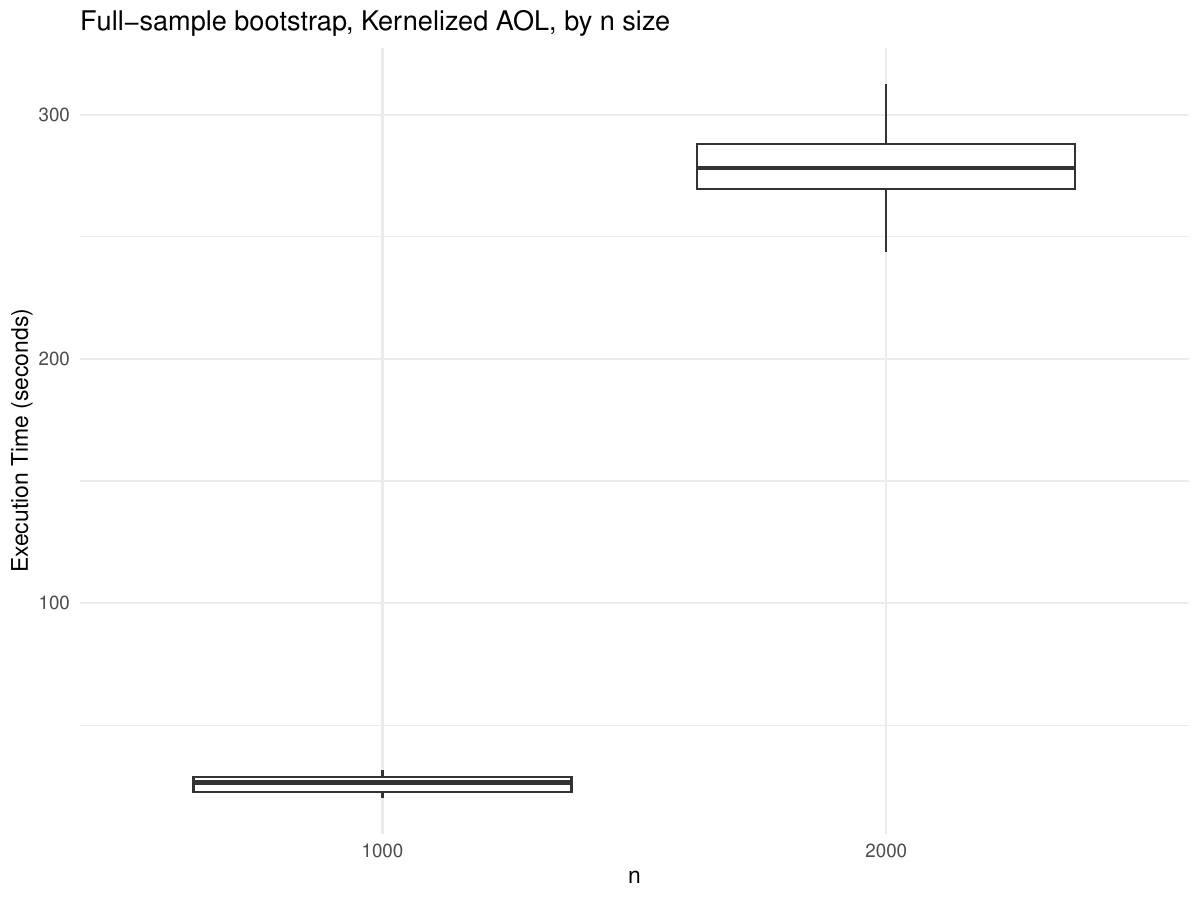}
    \caption{Timing results from \timereps{} replications of the full-sample bootstrap for policy learning with Kernelized AOL}
    \label{fig:timing_scaling_policy}
\end{figure}

\subsection{Robustness to Covariate Overlap Violations and Nuisance Misspecification}
\label{sec:robustness}

The simulations above use data-generating processes with reasonable covariate overlap and correctly specified nuisance models. We now test each of these assumptions in turn for the DML estimator, using $n = 5{,}000$ and $1{,}000$ replications throughout.

\textbf{Covariate overlap.} We vary the degree of covariate imbalance in the treatment-assignment mechanism from mild to severe and compare cBLB coverage against the full-sample bootstrap. We ran cBLB with 5, 10, and 15 subsets; at fixed $n$, more subsets means a smaller subset size $b$. Under mild overlap the number of subsets does not matter (coverage of 0.950, 0.947, and 0.949 for 5, 10, and 15 subsets, respectively), but under severe overlap coverage drops as subsets increase and $b$ shrinks (0.823, 0.750, and 0.709). Table~\ref{tab:overlap} reports the 15-subset case, where the effect is most visible. Both methods' coverage declines as overlap worsens, but cBLB's declines substantially faster: at the most severe setting, cBLB coverage falls to 70.9\% versus 91.8\% for the full-sample bootstrap.

Over 30 replications at $n = 5{,}000$, we fit a cross-fitted kernel-SVM propensity model once on the full sample and once on one random subsample of size $b = 333$, the subset size actually used with 15 subsets, and compute the Kish effective sample size, $\mathrm{ESS} = (\sum_i w_i)^2 / \sum_i w_i^2$, where $w_i$ is the inverse-propensity weight for unit $i$. Averaged across replications, ESS is 206.5 for the full sample but 86.6 for the subset. With only $b = 333$ observations to fit the propensity model, a single cBLB subset has far less information to work with than the full sample. Per-subset propensity and treatment-effect estimates become noisy, and resampling within one small subset does not capture this added variability across subsets.

This is consistent with cBLB's underlying theory (Theorem~\ref{thm:norefit}), which establishes asymptotically correct coverage as the subset size $b \to \infty$. The degradation here reflects a finite-sample effect at the practical, fixed subset sizes used in this paper. A finite-sample correction, for example inflating subset-level variance estimates to account for small-$b$ variation in the nuisance fits is left for future work.

\begin{table}[t]
\centering
\caption{\label{tab:overlap}Coverage under increasing covariate overlap violation, DML with kernel SVM ($n=5{,}000$, 1{,}000 replications)}
\begin{tabular}{lll}
\toprule
Overlap violation & cBLB coverage & Full-sample bootstrap coverage\\
\midrule
Mild & 0.949 & 0.949\\
Moderate & 0.850 & 0.928\\
Severe & 0.709 & 0.918\\
\bottomrule
\end{tabular}
\end{table}

\textbf{Nuisance misspecification.} We next deliberately misspecify the DML nuisance models. Treatment and outcome are generated from true confounders $Z_1, Z_2 \sim N(0,1)$, but the linear-kernel propensity and outcome regressions are fit on the transformed covariates $X_1 = \exp(Z_1/2)$ and $X_2 = Z_2/(1+\exp(Z_1)) + 10$ instead of $Z_1$ and $Z_2$ themselves. Table~\ref{tab:misspec} reports coverage for cBLB and the full-sample bootstrap under this misspecification. Coverage collapses for both: 6.1--7.2\% for cBLB depending on the number of subsets, and 13.0\% for the full-sample bootstrap. The degradation is therefore mainly a property of the underlying DML estimator under this level of misspecification, with cBLB undercovering further. Figure~\ref{fig:zip_robustness_cblb} shows the resulting confidence intervals for cBLB under both severe overlap violation and misspecification, and Figure~\ref{fig:zip_robustness_full} shows the same for the full-sample bootstrap.

\begin{table}[t]
\centering
\caption{\label{tab:misspec}Coverage under misspecified DML nuisance models ($n=5{,}000$, 1{,}000 replications)}
\begin{tabular}{ll}
\toprule
Setting & Coverage\\
\midrule
cBLB (5 subsets) & 0.072\\
cBLB (10 subsets) & 0.061\\
cBLB (15 subsets) & 0.066\\
Full-sample bootstrap & 0.130\\
\bottomrule
\end{tabular}
\end{table}

\begin{figure}
    \centering
    \includegraphics[scale = 0.75]{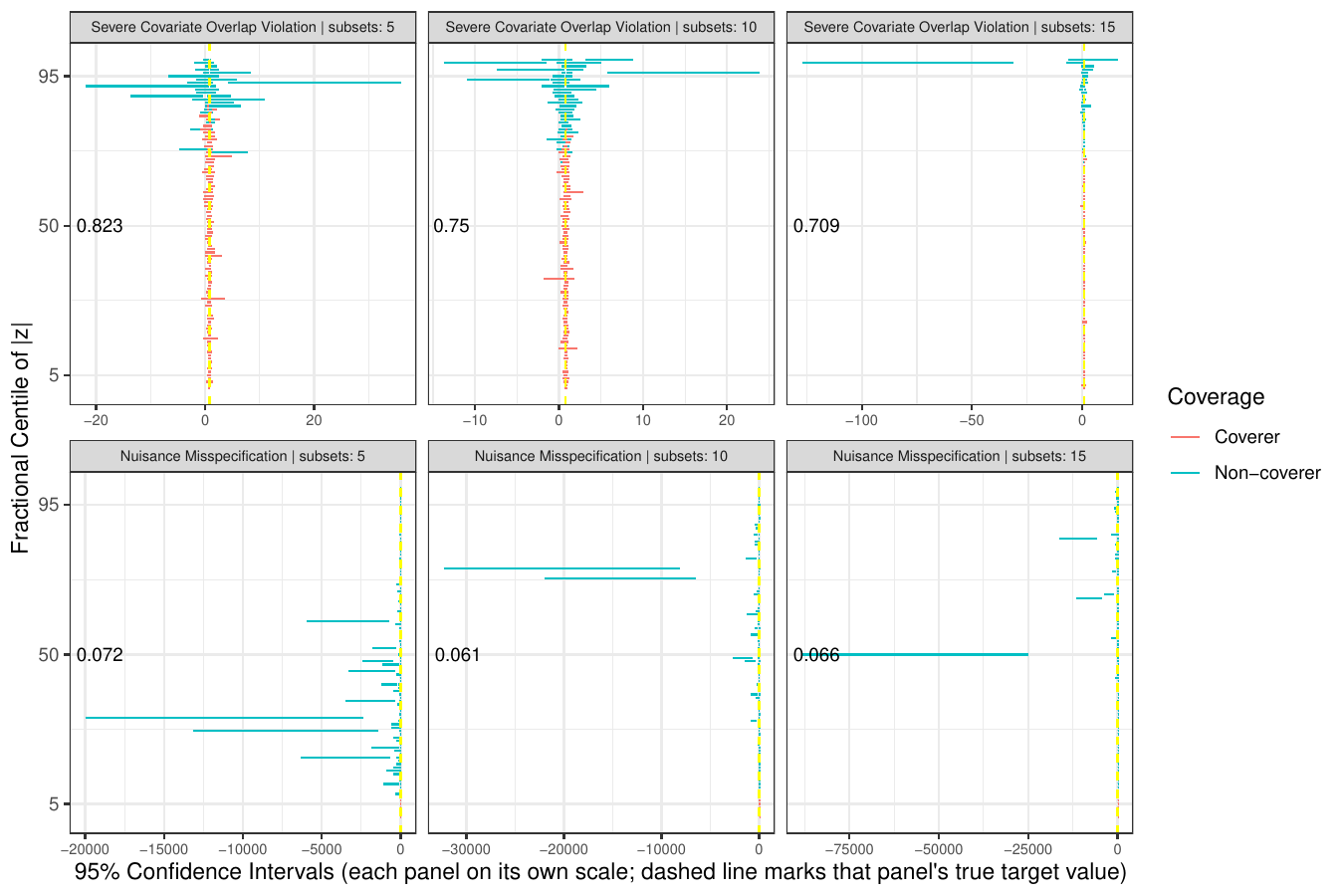}
    \caption{Confidence intervals for the ATE under severe covariate overlap violation and under nuisance misspecification, DML using SVM and cross-fitting, cBLB algorithm (panels by subset count).}
    \label{fig:zip_robustness_cblb}
\end{figure}

\begin{figure}
    \centering
    \includegraphics[scale = 0.75]{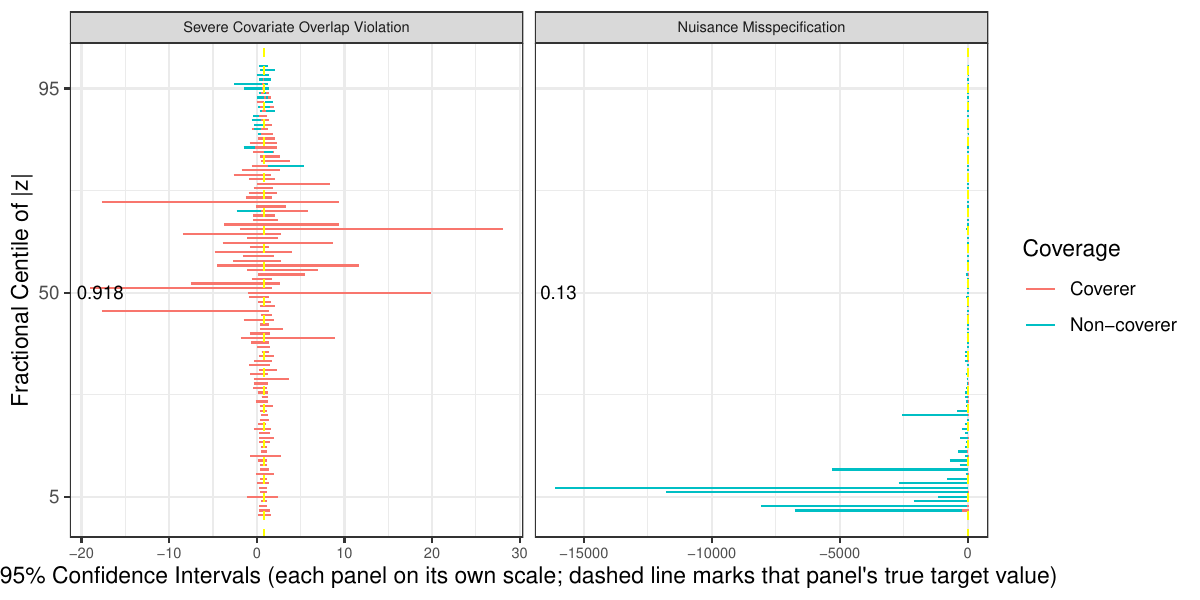}
    \caption{Confidence intervals for the ATE under severe covariate overlap violation and under nuisance misspecification, DML using SVM and cross-fitting, full-sample bootstrap.}
    \label{fig:zip_robustness_full}
\end{figure}

\subsection{Practical Considerations}
\label{sec:pract_guidelines}

\subsubsection{Size of the subsets $b$} The selection of subset size and the number of such subsets are closely interrelated and influence one other. Generally, a smaller number of observations within each subset necessitates a greater number of subsets to yield satisfactory results. In the foundational BLB paper, the authors recommend a formula for determining subset sizes $b$ relying on a parameter $\gamma$ such that $b = n^{\gamma}$ \citep[p. 20]{kleiner2014scalable}. After extensive simulations, they conclude that a $\gamma$ value of 0.7 often serves as ``an effective and sensible choice'' in a variety of contexts \citep[p. 20]{kleiner2014scalable}. In our previous work \cite{kosko2023fast}, we found that the appropriate $\gamma$ changes by estimating method. In this paper, we recommend deciding on the subset size based on domain knowledge, trading off the computational complexity of the algorithm against the minimum number of observations needed to accurately optimize or estimate. Each subset size was chosen to be in the several hundreds or low thousands of observations.  

\subsubsection{Number of subsets $s$} 
The number of subsets required depends on the size of the subset chosen. In the original paper, \cite{kleiner2014scalable} found that 1-2 subsets could be sufficient for a $\gamma = 0.9$ and recommended 10-20 subsets for $\gamma = 0.7$. \cite{kleiner2014scalable} prove that the BLB enjoys the same higher order correctness as the bootstrap as long as we choose disjoint subsets and the number of subsets is approximately $\frac{n}{b}$. We retain this design in cBLB, partitioning the data into disjoint subsets with $s \approx n/b$. Thus, once the subset size or number of subsets is decided, the other is given automatically by the formula $b \approx \frac{n}{s}$. 

\subsubsection{Estimands} Our simulations estimate (1) the ATE and (2) the optimal value of an outcome from the optimal regime, but we can in fact calculate any type of estimand with the cBLB. It is agnostic to estimand.

\subsubsection{Constructing confidence intervals and hypothesis tests} We can use any method to construct confidence intervals, including percentile and normal \citep{efron1994introduction, davison1997bootstrap}. In this simulation, we use percentile.

\section{Application to the 2023 U.S. Natality Data}\label{sec:application}

\subsection{Study Population}\label{sec:studypop}

We conducted a case study using the 2023 U.S. natality public-use microdata 2023 Natality Public Use File from the National Vital Statistics System (NVSS), compiled by the National Center for Health Statistics (NCHS) from U.S. birth certificates submitted by state and territorial vital registration jurisdictions \citep{nchs_natality2023_userguide}. 

The 2023 data %microdata are distributed as fixed-width records of length 1330 characters; the 2023 file 
contains 3{,}605{,}081 U.S. birth records  \citep{nchs_natality2023_userguide}. NCHS reports that more than 99\% of births occurring in the United States in 2023 were registered, and emphasizes that item-level completeness varies across variables and jurisdictions \citep{nchs_natality2023_userguide}.

The dataset includes detailed measures of maternal sociodemographic characteristics (e.g., age, race/ethnicity, education, marital status), paternal characteristics, prenatal care utilization (e.g., month prenatal care began, number of visits), maternal behaviors and health (e.g., cigarette use before and during pregnancy, pre-pregnancy BMI), labor/delivery characteristics, and infant outcomes\citep{nchs_natality2023_userguide}.

The analytic cohort was restricted to singleton live births with plausible birthweights (350--6000 g) and complete information on maternal and pregnancy covariates. Of the 3{,}605{,}081 records in the 2023 file, we excluded 113{,}346 plural births, leaving 3{,}491{,}735 singleton births, and a further 4{,}071 singleton births with birthweight missing (2{,}751), below 350 g (1{,}179) or above 6{,}000 g (141), leaving 3{,}487{,}664. We then excluded the 506{,}925 of these births missing at least one confounder (maternal age, education, race, Hispanic origin, marital status, parity, month prenatal care began, or payment source), which yields the final analytic sample of 2{,}980{,}739 births used throughout this section. Marital status accounts for most of this loss: it is missing for 389{,}278 of the 3{,}487{,}664 singleton births and is the only missing confounder for 344{,}129 of them. Each remaining confounder is missing for fewer than 65{,}000 births (month prenatal care began, 64{,}166; education, 61{,}787; Hispanic origin, 31{,}153; payment source, 24{,}952; parity, 8{,}698; age and race, none). Smoking status was not an exclusion criterion: cigarettes per day recorded as unknown in a trimester were treated as zero for that trimester, so the 7{,}540 births in the analytic sample with unknown cigarette use in all three trimesters are classified as non-smokers. We treat this as a limitation of the present analytic sample, not a data quality issue with the NVSS file itself. The primary outcome was infant birthweight in grams, and the exposure was an indicator for any maternal cigarette use during pregnancy, defined as positive self-reported cigarettes per day in any trimester. To address confounding, we adjusted for maternal age, education, race, Hispanic origin, marital status, parity, prenatal care timing, and payment source. We learned an individualized treatment regime using Kernelized Augmented Outcome-Weighted Learning and then estimated the optimal value. For kernelized AOL we used a linear kernel with $\lambda = 0.01$ and $s = 2{,}500$ subsets ($b = 1{,}192$). Within each subset, the propensity score was estimated by logistic regression on the same confounders and clipped to $[0.01, 0.99]$, identically in every subset, and the AOL objective was minimized by L-BFGS-B using its closed-form gradient. We estimated ATE using both Kernel Minimax Weights with kernels as in our simulations and DML with a linear kernel. 

Table \ref{tab:baseline_smoking} shows some baseline characteristics of the sample by smoking status. Table \ref{tab:baseline_covariates} shows the distribution of important covariates by smoking status. 

\begin{table}[t]
\centering
\caption{\label{tab:baseline_smoking}Baseline characteristics by smoking status (NVSS)}
\centering
\begin{tabular}[t]{llllll}
\toprule
Smoking status & N & Mean birthweight & SD birthweight & Mean age & SD age\\
\midrule
Non-smoker & 2,882,907 & 3279.397 & 552.137 & 29.401 & 5.797\\
Smoker &    97,832 & 3056.878 & 589.699 & 29.277 & 5.623\\
\bottomrule
\end{tabular}
\end{table}

\begin{table}[t]
\centering
\caption{\label{tab:baseline_covariates}Covariate distributions by smoking status (NVSS)}
\centering
\begin{tabular}[t]{l >{\RaggedRight\arraybackslash}p{6cm} ll}
\toprule
Covariate & Level & Non-smoker & Smoker\\
\midrule
Age\_group & $<$ 20 &   120,609 (4.2\%) &     2,955 (3.0\%)\\
 & 20-24 &   510,662 (17.7\%) &    19,070 (19.5\%)\\
 & 25-29 &   805,023 (27.9\%) &    28,421 (29.1\%)\\
 & 30-34 &   872,505 (30.3\%) &    28,933 (29.6\%)\\
 & 35-39 &   464,653 (16.1\%) &    14,940 (15.3\%)\\
 & 40+ &   109,455 (3.8\%) &     3,513 (3.6\%)\\
Education & 8th grade or less &    97,117 (3.4\%) &     2,035 (2.1\%)\\
 & 9th through 12th grade with no diploma &   218,887 (7.6\%) &    21,702 (22.2\%)\\
 & Associate degree (AA, AS) &   246,091 (8.5\%) &     5,256 (5.4\%)\\
 & Bachelor’s degree (BA, AB, BS) &   661,003 (22.9\%) &     2,674 (2.7\%)\\
 & Doctorate (PhD, EdD) or Professional Degree (MD, DDS, DVM, LLB, JD) &    91,746 (3.2\%) &        79 (0.1\%)\\
 & High school graduate or GED completed &   764,892 (26.5\%) &    44,812 (45.8\%)\\
 & Master’s degree (MA, MS, MEng, MEd, MSW, MBA) &   312,907 (10.9\%) &       542 (0.6\%)\\
 & Some college credit, but not a degree &   490,264 (17.0\%) &    20,732 (21.2\%)\\
Precare & No prenatal care &    63,948 (2.2\%) &     6,631 (6.8\%)\\
Race  & AIAN (only) &    28,036 (1.0\%) &     2,323 (2.4\%)\\
  & Asian (only) &   155,581 (5.4\%) &       392 (0.4\%)\\
  & Black (only) &   475,092 (16.5\%) &    12,599 (12.9\%)\\
  & More than one race &    80,133 (2.8\%) &     4,486 (4.6\%)\\
  & NHOPI (only) &    10,863 (0.4\%) &       173 (0.2\%)\\
  & White (only) & 2,133,202 (74.0\%) &    77,859 (79.6\%)\\
\bottomrule
\end{tabular}
\begin{tablenotes}
    \item \tiny GED, General Educational Development; AIAN, American Indian or Alaska Native; NHOPI, Native Hawaiian or Other Pacific Islander.
\end{tablenotes}
\end{table}

\subsection{Results}\label{sec:whiresults}

Table~\ref{tab:casestudy} presents the causal estimates, confidence intervals, and computation times for each method using the cBLB procedure. We organize our discussion around three aspects: the substantive estimates, their alignment with the existing literature, and the computational performance of the cBLB relative to the full-sample bootstrap. For kernelized AOL and kernel minimax weights, we used $s = 2{,}500$ subsets of size $b = 1{,}192$; for DML with SVM, we used $s = 300$ subsets of size $b = 9{,}935$, with $s \approx n/b$ (Section~\ref{sec:pract_guidelines}).

\subsubsection{Average treatment effect estimates} The kernel minimax weights estimator yields an estimated ATE of $-219.3$ grams (95\% CI: $-223.3$, $-215.4$), while the DML with SVM estimator produces an estimate of $-185.1$ grams (95\% CI: $-189.1$, $-181.0$). Both estimates are negative, statistically significant, and fall within the range of 150--250 grams typically reported in the epidemiological literature on smoking and birthweight \citep{hellerstedt1995effects, wisborg2001maternal, cattaneo2010efficient}. The confidence intervals are narrow, reflecting the precision afforded by the large sample size, with standard errors of 1.93 and 2.00 grams, respectively. As a non-kernel benchmark, we also estimate the ATE using AIPW with a logistic propensity score model and arm-specific linear outcome regressions, fit once on the full sample with a closed-form influence-curve standard error (Table~\ref{tab:casestudy}); this yields an estimate of $-133.0$ grams (95\% CI: $-227.4$, $-38.7$), also negative but with a substantially wider confidence interval and a point estimate closer to zero than either kernel-based estimator. The two methods' estimated mean birthweight under no smoking, $E[Y(0)]$, differ by only 0.40 grams (3278.55 vs.\ 3278.15 g); nearly all of the 34-gram ATE discrepancy is instead in their implied mean birthweight under smoking, $E[Y(1)]$ (3059.22 vs.\ 3093.09 g). Smokers make up only 3.3\% of the analytic sample (Table~\ref{tab:baseline_smoking}), so the treated-arm quantities are estimated from far less data. Despite this difference, the qualitative conclusions are consistent across both methods.

\subsubsection{Optimal regime value} The kernelized AOL estimator yields an estimated optimal value of 3{,}283.5 grams (95\% CI: 3{,}271.7, 3{,}295.3). This quantity is on a fundamentally different scale from the ATE estimates: it represents the expected birthweight when each mother is assigned to the treatment recommended by the optimal regime. The interval contains our estimates of the mean birthweight had no mother smoked, from DML with SVM and kernel minimax weights (3{,}278.1 and 3{,}278.5, respectively; Table~\ref{tab:casestudy}), and is consistent with population-level mean birthweights reported for nonsmoking mothers in the United States \citep{nchs_natality2023_userguide}. The optimal value is therefore indistinguishable from that produced by universal abstinence: given the well-established harm of smoking, no assignment of smoking improves on recommending that no one smoke during pregnancy.

\begin{table}[t]
\centering
\caption{\label{tab:casestudy} ATE estimates and mean outcome under the optimal policy for the NVSS dataset}
\begin{tabular}{lllllll}
\toprule
Setting & $E(Y(0))$\tnote{a} & Lower CI & Upper CI & Estim. & S.E. & Time (seconds)\\
\midrule
Kernel minimax weights & 3278.546 & -223.260 & -215.404 & -219.324 & 1.929 & 50198.821\\
DML with SVM & 3278.147 & -189.123 & -180.980 & -185.053 & 1.997 & 38733.731\\
Non-kernel benchmark (AIPW, logistic PS)\tnote{b} & -- & -227.372 & -38.661 & -133.016 & 48.141\tnote{c} & --\\
Kernelized AOL & -- & 3271.741 & 3295.344 & 3283.480 & 5.795 & 1290.409\\
\bottomrule
\end{tabular}
\begin{tablenotes}
    \item[a] \tiny Under standard assumptions $E[Y(0)] = E[ E[Y | A=0,X] ]$
    \item[b] \tiny AIPW with a logistic propensity score model and arm-specific linear outcome regressions, fit once on the full sample; standard error from the closed-form influence curve (no cross-fitting or bootstrap)
    \item[c] \tiny With estimated propensities clipped to $[0.01, 0.99]$, as for kernelized AOL, the estimate is $-151.5$ grams with standard error $2.07$ grams.
\end{tablenotes}
\end{table}

\subsubsection{Computational performance} A central motivation for this work is the computational infeasibility of standard bootstrap inference with kernel-based estimators at this scale. In our application, the full-sample bootstrap failed outright with an out-of-memory error for both the kernel minimax weights and kernelized AOL estimators, and ran for approximately one week without producing results for the DML estimator. By contrast, the cBLB procedure successfully produced confidence intervals for all three methods, with computation times of approximately 21.5 minutes for kernelized AOL, 10.8 hours for DML with SVM, and 13.9 hours for kernel minimax weights. For kernelized AOL, each subset requires only a $b \times b$ Gram matrix ($b = 1{,}192$) and a single L-BFGS-B fit, which we solve with the closed-form gradient of the AOL objective rather than numerical differentiation over its $b+1$ parameters. As in the simulations, the Gaussian process fits for kernel minimax weights used single-threaded linear algebra (\texttt{OPENBLAS\_NUM\_THREADS=1}). These computation times were obtained without parallelization across subsets; exploiting the embarrassingly parallel structure of cBLB would reduce runtimes proportionally to the number of available cores.

\section{Conclusions}
\label{sec:conclusion}
This paper addresses the pressing computational challenges of conducting valid inference with kernel-based causal estimators on large real-world datasets. Building on recent work on the causal bag of little bootstraps (cBLB), we extend this scalable bootstrap framework to accommodate a wide class of kernel-based methods used in causal inference. These include kernel minimax weights in reproducing kernel Hilbert spaces (RKHS), doubly robust estimators within the double machine learning (DML) framework, and optimal policy learning estimators based on augmented outcome-weighted learning (AOL).

Through extensive simulations, we demonstrate that the proposed algorithm achieves nominal coverage for confidence intervals while offering dramatic improvements in computational efficiency compared to the full-sample bootstrap. These gains hold across a range of causal estimands, including average treatment effects and optimal policy values, and across multiple domains with varying levels of model complexity. The results also suggest that cBLB can deliver reliable inference even when nuisance components are estimated using complex machine learning models.

From a methodological standpoint, our work contributes a flexible resampling framework that is agnostic to the causal estimand and can be adapted to other high-dimensional or kernel-based methods, including those based on Gaussian processes, support vector machines, and RKHS-regularized estimators. This flexibility is crucial for practitioners seeking scalable uncertainty quantification tools that maintain statistical rigor in the presence of rich covariate information and complex treatment assignment mechanisms.

However, several limitations remain. First, although our method avoids repeated nuisance model fitting in the bootstrap loop, this relies on the assumption that re-using fitted nuisance functions provides adequate finite-sample inference. While our simulations and theory suggest this is reasonable in practice, there may be situations in which this assumption may fail. Second, the performance of our algorithm depends on subset size and the number of bags, which require tuning and may depend on the data generating process or estimand; in particular, under limited covariate overlap, small subsets can produce substantial undercoverage (Section~\ref{sec:robustness}). Third, the current implementation does not include automatic selection of kernel parameters or regularization penalties, which may limit performance in some applications. Fourth, for kernelized AOL the treatment rule is learned once within each subset and held fixed during resampling, so we do not provide theoretical guarantees for the distribution of the learned decision rule itself, and for this estimator Theorem~\ref{thm:norefit} provides only a heuristic justification for coverage of the optimal value (Section~\ref{subsec:properties}).

Future research should explore adaptive schemes for choosing hyperparameters within the cBLB, including automated tuning of subset size, kernel bandwidths, and regularization parameters. Additionally, further theoretical work is needed to characterize the finite-sample coverage properties of BLB-style methods. Lastly, extending the cBLB to multi-valued or continuous treatment settings and time-varying exposures would further increase its practical relevance, especially in biomedical and policy evaluation contexts.

Overall, our results show that scalable inference for kernel-based causal estimators is not only feasible but also statistically valid in a wide range of realistic scenarios. We believe that the proposed approach provides a valuable blueprint for integrating advanced machine learning estimators into causal analysis pipelines at scale.

% \section*{Acknowledgments}

% The authors used ChatGPT 5.1 during the preparation of this work to assist with generating \texttt{R} code for figure production and debugging code. All AI-generated code was reviewed and validated by the authors.

% \section*{Data Availability Statement}
% The data that support the findings of this study are publicly available from the National Center for Health Statistics through the Vital Statistics Online Data Portal (\url{https://www.cdc.gov/nchs/data_access/vitalstatsonline.htm}). The analytic dataset used in this study was derived from the 2023 Natality Public Use File and associated public-domain documentation/resources, including the National Bureau of Economic Research Vital Statistics Natality Birth Data page 
% (\url{https://www.nber.org/research/data/vital-statistics-natality-birth-data}) and the 2023 natality dictionary file (\url{https://data.nber.org/nvss/natality/programs/dct/natality2023.dct}).

\backmatter

% \bmhead{Supplementary information}

% If your article has accompanying supplementary file/s please state so here. 

% Authors reporting data from electrophoretic gels and blots should supply the full unprocessed scans for key as part of their Supplementary information. This may be requested by the editorial team/s if it is missing.

% Please refer to Journal-level guidance for any specific requirements.

% \bmhead{Acknowledgements}

% Acknowledgements are not compulsory. Where included they should be brief. Grant or contribution numbers may be acknowledged.

% Please refer to Journal-level guidance for any specific requirements.

\section*{Declarations}
\begin{itemize}
\item \textbf{Funding:} This article is based upon work supported by The Amazon Web Services (AWS) funding for ``AI/ML for Identifying Social Determinants of Health'', the National Science Foundation under Grant No 2306556, and the National Institute of Health Grant No 1R01AI197146-01.
\item \textbf{Conflict of interest/Competing interests:} Not applicable.
\item \textbf{Ethics approval and consent to participate:} Not applicable.
\item \textbf{Consent for publication:} Not applicable.
\item \textbf{Data availability:} The data that support the findings of this study are publicly available from the National Center for Health Statistics through the Vital Statistics Online Data Portal (\url{https://www.cdc.gov/nchs/data_access/vitalstatsonline.htm}). The analytic dataset used in this study was derived from the 2023 Natality Public Use File and associated public-domain documentation/resources, including the National Bureau of Economic Research Vital Statistics Natality Birth Data page (\url{https://www.nber.org/research/data/vital-statistics-natality-birth-data}) and the 2023 natality dictionary file (\url{https://data.nber.org/nvss/natality/programs/dct/natality2023.dct}).
\item \textbf{Code availability:} The \texttt{R} code for the algorithm is available at \url{https://github.com/mdk31/kernel-blb}.
\item \textbf{Author contribution:} Not applicable.
\end{itemize}

\begin{appendices}
\theoremstyle{thmstyleone}%
\newtheorem{lemmaA}{Lemma}[section]%
\newtheorem{propA}[lemmaA]{Proposition}%
\theoremstyle{thmstyletwo}%
\newtheorem{remarkA}[lemmaA]{Remark}%

\section{Proof of Theorem~\ref{thm:norefit}}\label{app:proof}

\paragraph{Proof roadmap.}
The proof has three steps. First, under \textbf{frozen subset contributions},
assumption~(ii) is used to show that the frozen contributions preserve the
first-order centering and variance of the oracle influence-function
contributions. Second, under the \textbf{conditional multinomial
approximation}, assumption~(iii) yields a Gaussian approximation for the
cBLB replicate conditional on the frozen contributions. Third, under
\textbf{aggregation and coverage}, the subset-level approximations are
combined using the BLB aggregation condition in
Theorem~\ref{thm:norefit} to recover the full-sample first-order law
identified by assumption~(i), and hence asymptotically valid confidence
intervals.

\begin{proof}

\paragraph{Frozen subset contributions.}
Let
\[
\phi_i=\theta_0+\psi(Z_i),
\qquad
v_k^2=\operatorname{Var}_{P_{b,k}}\{\psi(Z_i)\}.
\]
By assumption~(ii), $\hat\theta_{i,k}=\phi_i+r_{i,k}$, and therefore
\[
\hat\theta_k-\theta_0
=
\frac1b\sum_{i\in I_k}\psi(Z_i)
+
\frac1b\sum_{i\in I_k}r_{i,k}
=
\frac1b\sum_{i\in I_k}\psi(Z_i)
+
o_p(b^{-1/2}).
\]
Thus, the frozen subset estimator has the same first-order centering as
the corresponding oracle influence-function average.

For the variance,
\[
\sigma_k^2
=
\operatorname{Var}_{P_{b,k}}
\{\psi(Z_i)+r_{i,k}\}.
\]
Hence, by Cauchy--Schwarz,
\[
\begin{aligned}
|\sigma_k^2-v_k^2|
&\leq
\operatorname{Var}_{P_{b,k}}(r_{i,k})
+
2\left|
\operatorname{Cov}_{P_{b,k}}\{\psi(Z_i),r_{i,k}\}
\right| \\
&\leq
\frac1b\sum_{i\in I_k}r_{i,k}^2
+
2\left\{
P_{b,k}\psi^2
\left(
\frac1b\sum_{i\in I_k}r_{i,k}^2
\right)
\right\}^{1/2}
=
o_p(1),
\end{aligned}
\]
where the last equality follows from assumption~(ii) and
$\psi\in L_2(P)$. Since $b\to\infty$, the law of large numbers gives
\[
v_k^2
\xrightarrow{p}
\sigma^2
:=
\operatorname{Var}\{\psi(Z)\}
\]
for every fixed $k$, and hence $\sigma_k^2\xrightarrow{p}\sigma^2$.

To obtain the across-subset statement, for any $\varepsilon>0$ let
\[
A_{k,n}
=
\mathbf 1\{|\sigma_k^2-\sigma^2|>\varepsilon\}.
\]
Because the observations are i.i.d.\ and the equal-size partition is
formed independently of their values, the subsets are identically
distributed. Therefore,
\[
E\left(
\frac1s\sum_{k=1}^s A_{k,n}
\right)
=
P(|\sigma_1^2-\sigma^2|>\varepsilon)
\longrightarrow0.
\]
Markov's inequality then gives
\[
\frac1s\sum_{k=1}^s
\mathbf 1\{|\sigma_k^2-\sigma^2|>\varepsilon\}
\xrightarrow{p}0,
\]
which establishes the second across-subset variance formulation stated
in Theorem~\ref{thm:norefit}.

\paragraph{Conditional multinomial approximation.}
Conditional on the frozen contributions,
$\hat\theta_{n,k}^*$ has the same distribution as the mean of $n$
i.i.d.\ draws from the empirical distribution that assigns probability
$1/b$ to each $\hat\theta_{i,k}$. Thus,
\[
E^*(\hat\theta_{n,k}^*)=\hat\theta_k,
\qquad
\operatorname{Var}^*\{\sqrt n(\hat\theta_{n,k}^*-\hat\theta_k)\}
=
\sigma_k^2.
\]

If $\sigma_k^2=0$, the conditional bootstrap distribution is degenerate
at $\hat\theta_k$ and the conclusion is immediate. We therefore consider
$\sigma_k^2>0$. Let
$x_{i,k}=\hat\theta_{i,k}-\hat\theta_k$. The conditional
Berry--Esseen inequality gives
\[
\sup_t
\left|
P^*\left\{
\frac{\sqrt n(\hat\theta_{n,k}^*-\hat\theta_k)}
{\sigma_k}
\leq t
\right\}
-
\Phi(t)
\right|
\leq
\frac{C}{\sqrt n}
\frac{P_{b,k}|x_{i,k}|^3}{\sigma_k^3}.
\]
Since $P_{b,k}x_{i,k}^2=\sigma_k^2$,
\[
P_{b,k}|x_{i,k}|^3
\leq
\max_{i\in I_k}|x_{i,k}|\,\sigma_k^2
\leq
\sqrt b\,\sigma_k^3,
\]
and consequently
\[
\sup_t
\left|
P^*\left\{
\frac{\sqrt n(\hat\theta_{n,k}^*-\hat\theta_k)}
{\sigma_k}
\leq t
\right\}
-
\Phi(t)
\right|
\leq
C\sqrt{\frac bn}
\longrightarrow0
\]
by assumption~(iii). Therefore, conditionally on the subset,
\[
\frac{\sqrt n(\hat\theta_{n,k}^*-\hat\theta_k)}
{\sigma_k}
\Rightarrow^*
N(0,1).
\]
Equivalently, the conditional Kolmogorov distance between the law of
$\sqrt n(\hat\theta_{n,k}^*-\hat\theta_k)$ and
$N(0,\sigma_k^2)$ converges to zero.

\paragraph{Aggregation and coverage.}
Let $\bar\theta=s^{-1}\sum_{k=1}^s\hat\theta_k$ and
$\bar\sigma=s^{-1}\sum_{k=1}^s\sigma_k$, and let $z_\alpha$ denote the
$\alpha$-quantile of $N(0,1)$. The conditional Berry--Esseen bound above
is $C\sqrt{b/n}$ for every subset, so the conditional $\alpha$-quantile of
$\hat\theta^*_{n,k}$ equals
$\hat\theta_k+n^{-1/2}\sigma_k(z_\alpha+\eta_{n,k})$ with
$|\eta_{n,k}|\le\eta_n$, where $\eta_n\to0$ does not depend on $k$.
Averaging over subsets, the lower and upper cBLB endpoints are
\[
L_{\mathrm{cBLB}}=\bar\theta+n^{-1/2}\{\bar\sigma z_{\alpha/2}+o_p(1)\},
\qquad
U_{\mathrm{cBLB}}=\bar\theta+n^{-1/2}\{\bar\sigma z_{1-\alpha/2}+o_p(1)\}.
\]
Under the conditions of Remark~\ref{rem:aggregation},
\[
\sqrt n(\bar\theta-\theta_0)
=
\frac1{\sqrt n}\sum_{i=1}^n\psi(Z_i)+o_p(1)
\Rightarrow
N(0,\sigma^2)
\quad\text{and}\quad
\bar\sigma\xrightarrow{p}\sigma .
\]
Since $L_{\mathrm{cBLB}}\le\theta_0\le U_{\mathrm{cBLB}}$ if and only if
$\bar\sigma z_{\alpha/2}+o_p(1)\le\sqrt n(\bar\theta-\theta_0)\le\bar\sigma z_{1-\alpha/2}+o_p(1)$,
Slutsky's lemma gives
\[
P\left\{
L_{\mathrm{cBLB}}
\leq
\theta_0
\leq
U_{\mathrm{cBLB}}
\right\}
\longrightarrow
\Phi(z_{1-\alpha/2})-\Phi(z_{\alpha/2})
=
1-\alpha .
\]
Finally, Algorithm~\ref{alg:cblb} uses empirical quantiles of $r$
replicates. The conditional law of $\hat\theta^*_{n,k}$ has scale
$n^{-1/2}\sigma_k$, so with $r=r_n\to\infty$ the Monte Carlo error of each
empirical quantile is $o_p(n^{-1/2})$ and the conclusion is unchanged.
This proves the result.
\end{proof}

\begin{remarkA}[Sufficient conditions for aggregation]
\label{rem:aggregation}
The final coverage statement additionally invokes the BLB aggregation
condition in Theorem~\ref{thm:norefit}. One sufficient set of conditions
is
\[
\frac{1}{\sqrt n}
\sum_{k=1}^s\sum_{i\in I_k}r_{i,k}
=
o_p(1),
\qquad
\frac{1}{sb}
\sum_{k=1}^s\sum_{i\in I_k}r_{i,k}^2
=
o_p(1).
\]
Indeed, if $C_n=\cup_{k=1}^s I_k$, then
$|C_n|=sb=n-o(n)$, and the first condition gives
\[
\frac1s\sum_{k=1}^s\hat\theta_k-\theta_0
=
\frac1{sb}\sum_{i\in C_n}\psi(Z_i)
+
o_p(n^{-1/2})
=
\frac1n\sum_{i=1}^n\psi(Z_i)
+
o_p(n^{-1/2}),
\]
where the second equality follows from $|C_n^c|=o(n)$,
$\psi\in L_2(P)$, and independence of the partition from the data.

For the variance terms, write again
$v_k^2=\operatorname{Var}_{P_{b,k}}\{\psi(Z_i)\}$. The same
Cauchy--Schwarz argument as in the proof yields
\[
\begin{aligned}
\frac1s\sum_{k=1}^s|\sigma_k^2-v_k^2|
&\leq
\frac1{sb}
\sum_{k=1}^s\sum_{i\in I_k}r_{i,k}^2 \\
&\quad+
2\left\{
\left(
\frac1{sb}\sum_{i\in C_n}\psi(Z_i)^2
\right)
\left(
\frac1{sb}
\sum_{k=1}^s\sum_{i\in I_k}r_{i,k}^2
\right)
\right\}^{1/2}
=
o_p(1).
\end{aligned}
\]
Moreover,
\[
\frac1s\sum_{k=1}^s v_k^2
=
\frac1{sb}\sum_{i\in C_n}\psi(Z_i)^2
-
\frac1s\sum_{k=1}^s
\left\{
\frac1b\sum_{i\in I_k}\psi(Z_i)
\right\}^2
\xrightarrow{p}
\sigma^2.
\]
The first term converges to
$E\{\psi(Z)^2\}=\sigma^2$ by the law of large numbers. For the second,
$E\{\psi(Z)\}=0$ implies
\[
E\left[
\left\{
\frac1b\sum_{i\in I_k}\psi(Z_i)
\right\}^2
\right]
=
\frac{\sigma^2}{b},
\]
so its average is $o_p(1)$ by Markov's inequality.

Combining the two displays above, $s^{-1}\sum_{k=1}^s\sigma_k^2\xrightarrow{p}\sigma^2$.
Together with the across-subset statement in Theorem~\ref{thm:norefit},
this gives $\bar\sigma=s^{-1}\sum_{k=1}^s\sigma_k\xrightarrow{p}\sigma$:
subsets with $|\sigma_k^2-\sigma^2|\le\varepsilon$ contribute at most
$\varepsilon/\sigma$ to $s^{-1}\sum_k|\sigma_k-\sigma|$, and the remaining
fraction $p_n\xrightarrow{p}0$ of subsets contributes at most
$\sigma p_n+\{p_n\,s^{-1}\sum_k\sigma_k^2\}^{1/2}\xrightarrow{p}0$ by
Cauchy--Schwarz. These are the two facts used in the final step of the
proof of Theorem~\ref{thm:norefit}.

Thus the aggregate conditions preserve both the first-order center and
variance across subsets. Combined with the uniform conditional
Berry--Esseen bound of order $O(\sqrt{b/n})$, they provide one explicit
set of sufficient conditions under which averaging the subset-level
bootstrap quantiles recovers the corresponding quantiles of the
full-sample first-order distribution.
\end{remarkA}
\end{appendices}

\bibliography{biblio}% common bib file

\end{document}